\documentclass{sig-alternate-05-2015}

\usepackage{amsfonts}
\usepackage{amsmath}
\usepackage{url}
\usepackage{amssymb}
\usepackage{amsmath}
\usepackage{algorithm,subfigure}
\usepackage{graphicx}
\usepackage[noend]{algorithmic}

\newtheorem{lemma}{Lemma}
\newtheorem{theorem}{Theorem}
\newtheorem{definition}{Definition}
\newtheorem{corollary}{Corollary}
\newtheorem{example}{Example}

\begin{document}

\title{When Social Influence Meets Item Inference}
\author{Hui-Ju Hung$^\star$, Hong-Han Shuai$^\dagger$, De-Nian Yang$%
	^\dagger$, Liang-Hao Huang$^\dagger$,\\ Wang-Chien Lee$^\star$, Jian Pei$^\mathsection$, Ming-Syan Chen$^\ddagger$ \\
	\affaddr{$^\star$The Pennsylvania State University, State College,
		Pennsylvania, USA} \\
	\affaddr{$^\dagger$Academia Sinica, Taipei, Taiwan} \\
	\affaddr{$^\mathsection$Simon Fraser University, Burnaby, Canada} \\
	\affaddr{$^\ddagger$National Taiwan University, Taipei, Taiwan} \\
}
\maketitle

\begin{abstract}
Research issues and data mining techniques for product recommendation and viral marketing have been widely studied. Existing works on seed selection in social networks do not take into account the effect of product recommendations in e-commerce stores. In this paper, we investigate the seed selection problem for viral marketing that considers both effects of social influence and item inference (for product recommendation). We develop a new model, \emph{Social Item Graph (SIG)}, that captures both effects in form of hyperedges. Accordingly, we formulate a seed selection problem, called \emph{Social Item Maximization Problem (SIMP)}, and prove the hardness of SIMP. We design an efficient algorithm with performance guarantee, called Hyperedge-Aware Greedy (HAG), for SIMP and develop a new index structure, called SIG-index, to accelerate the computation of diffusion process in HAG. Moreover, to construct realistic SIG models for SIMP, we develop a statistical inference based framework to learn the weights of hyperedges from data. Finally, we perform a comprehensive evaluation on our proposals with various baselines. Experimental result validates our ideas and demonstrates the effectiveness and efficiency of the proposed model and algorithms over baselines.

\end{abstract}

\section{Introduction}

\label{sec:_intro}


The ripple effect of social influence \cite{Bond12Nature} has been explored for
viral marketing via online social networks. Indeed,
studies show that customers tend to receive product information from
friends better than advertisements on traditional media \cite{Nail04}.
To explore the potential impact of social influence, many research studies on \emph{seed selection}, i.e., selecting a given
number of influential customers to maximize the spread of social
recommendation for a product, have been reported~\cite%
{Chen10KDD, Kempe03KDD}.\footnote{%
	All the top 5 online retailers, including Amazon, Staples, Apple, Walmart,
	and Dell, are equipped with sophisticated recommendation engines. They also
	support viral marketing by allowing users to share favorite products in
	Facebook.} 
However, these works do not take into account the effect of product recommendations in online e-commerce stores.
We argue that when a customer buys an item due to the social influence (e.g., via Facebook or Pinterest), there is a potential side effect due to the \textit{item inference} recommendations from stores.\footnote{%
	In this paper, we refer product/item recommendation based on
	associations among items inferred from purchase transactions as item inference recommendation.}
For example, when Alice buys a DVD of ``Star War'' due
to the recommendation from friends, she may also pick up the
original novel of the movie due to an in-store recommendation, which may in turn
trigger additional purchases of the novel among her friends. To the
best of our knowledge, this additional spread introduced by
the item inference recommendations has not been considered in
existing research on viral marketing.

\begin{figure}[t]
	\centering
	\includegraphics[width=3 in]{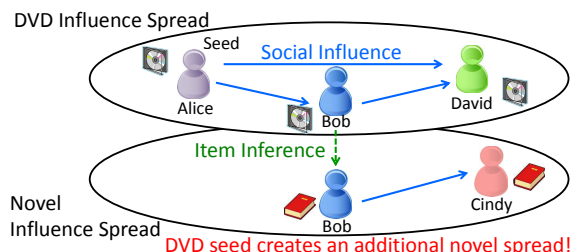} 
	\vspace{-10pt}
	\caption{A motivating example}
	\label{fig:motivation_example}
	\vspace{-2mm}
\end{figure}

Figure~\ref{fig:motivation_example} illustrates the above joint effects in a toy example
with two products and four customers, where a dash
arrow represents the association rule behind the item inference
recommendation, and a solid arrow denotes the social influence between
two friends upon a product. In the two separate
planes corresponding to DVD and novel, social influence is expected to take effect on promoting interests in (and
potential purchases of) the DVD and novel, respectively.
Meanwhile, the item inference recommendation by the e-commerce store is expected to trigger sales of additional
items. Note that the association rules behind item inference are derived without considering the ripple effect of social
influence. In the example, when
Bob buys the DVD, he may also buy the novel due to the item inference
recommendation. Moreover, he may influence Cindy to purchase
novel. However, the association rules behind item inference are derived without considering the ripple effect of social influence. On the other hand, to promote the movie DVD, Alice may be selected as a seed for a viral
marketing campaign, hoping to spread her influence to Bob and David to
trigger additional purchases of the DVD. Actually, due to the effect of item
inference recommendation, having Alice as a seed may additionally trigger
purchases of the novel by Bob and Cindy. This is a factor that existing seed
selection algorithms for viral marketing do not account for.

We argue that to select seeds for maximizing the spread of product information to a customer base (or maximizing the sale revenue of
products) in a viral marketing campaign, both effects of item inference
and social influence need to be considered. To incorporate both effects, we propose a new model, called \emph{Social Item Graph (SIG)} in form of hyperedges, for capturing ``purchase actions'' of customers on products and their potential influence to trigger other purchase actions. Different from the conventional approaches \cite{Chen10KDD,Kempe03KDD} that use links between customers to model social relationship (for viral marketing) and links between items to capture
the association (for item inference recommendation), SIG represents a
\emph{purchase action} as a node (denoted by a tuple of a customer and an
item), while using hyperedges among nodes to capture the influence spread
process used to predict customers' future purchases. Unlike the
previous influence propagation models \cite{Chen10KDD, Kempe03KDD}
consisting of only one kind of edges connecting two customers (in social
influence), the hyperedges in our model span across tuples of different
customers and items, capturing both effects of social influence and item
inference.

Based on SIG, we formulate the \emph{Social Item Maximization Problem} (SIMP) to find a seed set, which consists of selected products along with targeted customers, to maximize the total adoptions of products by customers. Note that SIMP takes multiple products into consideration and targets on maximizing the number of products
purchased by customers.\footnote{SIMP can be extended to a weighted version with different profits from each product. In this paper, we focus on maximizing the total sales.} SIMP is a very challenging problem, which does not have the submodularity property.
We prove that SIMP cannot be approximated within $n^{c}$ with any $c<1$, where $n$ is the number of nodes in SIMP,
i.e., SIMP is extremely difficult to approximate with a small ratio because the
best approximation ratio is almost $n$.\footnote{While there is no good solution quality guarantee for the worst case scenario, we empirically show that the algorithm we developed achieves total adoptions on average comparable to optimal results. }

To tackle SIMP, two challenges arise: 1) numerous combinations of
possible seed nodes, and 2) expensive on-line computation of influence
diffusion upon hyperedges. To address the first issue, we first introduce the \textit{Hyperedge-Aware Greedy (HAG)} algorithm, based on a unique property of hyperedges, i.e., a hyperedge requires all its source nodes to be activated in order to trigger the purchase action in its destination node. HAG selects multiple seeds  in each seed selection iteration to further activate more nodes via hyperedges.\footnote{A hyperedge requires all its source nodes to be activated to diffuse its influence to its destination node.}
To address the second issue, we exploit the structure of Frequent Pattern Tree (FP-tree) to develop \emph{SIG-index} as an compact representation of SIG in order to accelerate the computation of activation probabilities of nodes in online diffusion.

Moreover, to construct realistic SIG models for SIMP, we also develop a statistical inference based framework to learn the weights of hyperedges from logs of purchase actions. Identifying the hyperedges and estimating the corresponding weights are major challenges for constructing of a SIG due to data sparsity and unobservable activations.
To address these issues, we propose a novel framework that employs smoothed expectation and maximization algorithm (EMS) \cite{Silverman90}, to identify hyperedges and estimate their values by kernel smoothing.

Our contributions of this paper are summarized as follows.
\begin{itemize}
\vspace{-2pt}
\item We observe the deficiencies in existing techniques for item inference recommendation and seed selection and propose the \textit{Social Item Graph (SIG)} that captures both effects of social influence and item inference in prediction of potential purchase actions.
\vspace{-2pt}
\item Based on SIG, we formulate a new problem, called \textit{Social Item Maximization Problem (SIMP)}, to select the seed nodes for viral marketing that effectively facilitates the recommendations from both friends and stores simultaneously. In addition, we analyze the hardness of SIMP.
\vspace{-2pt}
\item We design an efficient algorithm with performance guarantee, called Hyperedge-Aware Greedy (HAG), and develop a new index structure, called SIG-index, to accelerate the computation of diffusion process in HAG.
\vspace{-12pt}
\item To construct realistic SIG models for SIMP, we develop a statistical inference based framework to learn the weights of hyperedges from data.
\vspace{-2pt}
\item We conduct a comprehensive evaluation on our proposals with various baselines. Experimental result validates our ideas and demonstrates the effectiveness and efficiency of the proposed model and algorithms over baselines.
\end{itemize}


The rest of this paper is organized as follows. Section \ref{sec:relatedwork} reviews the related work.
Section \ref{sec:social item graph} details the SIG model and its influence diffusion process.
Section \ref{sec:social_item_maximization} formulates SIMP and designs new algorithms to efficiently solve the problem.
Section \ref{sec:IdenHyperedge} describes our approach to construct the SIG.
Section~\ref{sec:experiment} reports our experiment results and Section~\ref{sec:conclusion}
concludes the paper.

\section{Related Work}

\label{sec:relatedwork}

To discover the associations among purchased items, frequent pattern mining
algorithms find items which frequently appear together in transactions~\cite{Agrawal93}.
Some variants, such as closed
frequent patterns mining \cite{Pasquier99}, maximal frequent pattern mining
\cite{Bayardo98}, 
have been studied. However, those existing works, focusing on unveiling the
common shopping behaviors of individuals, disregard the social influence
between customers~\cite{Wen10}. On the other hand, it has been pointed out that items recommended by item inference may have been introduced to users by  social diffusion~\cite{CIKM15}.
In this work, we develop a new model and a learning framework that consider both the social influence and item inference factors jointly to derive the association among
purchase actions of customers. In addition, we focus on seed selection for prevalent viral marketing by incorporating the effect of item inference.

With a great potential in business applications, social influence diffusion
in social networks has attracted extensive interests recently~\cite%
{Chen10KDD, Kempe03KDD}. Learning algorithms for estimating the social
influence strength between social customers have been developed~\cite%
{Goyal10WSDM, Saito08KES}. Based on models of social influence diffusion,
identifying the most influential customers (seed selection) is a widely studied problem~\cite%
{Chen10KDD, Kempe03KDD}. Precisely, those studies aim to find the best $k$
initial seed customers to target on in order to maximize the population of
potential customers who may adopt the new product. This seed selection
problem has been proved as NP-hard~\cite{Kempe03KDD}. Based on two influence
diffusion models, Independent Cascade (IC) and Linear Threshold (LT), Kempe
et al. propose a $1-1/e$ approximation greedy algorithm by exploring the
submodularity property under IC and LT~\cite{Kempe03KDD}. 
Some follow-up studies focus on improving the efficiency of the greedy algorithm using various spread estimation methods, e.g., MIA\cite{Chen10KDD} and TIM+\cite{Tang14SIGMOD}. However, without considering the existence of item inference, those algorithms are not applicable to SIMP. Besides the IC and LT model, Markov random field has been used to model social influence and calculate expected profits from viral marketing \cite{Domingos01}. Recently, Tang et al. proposed a Markov model based on ``confluence'', which estimates the total influence by combining different sources of conformity~\cite{Tang13}. However, these studies only consider the diffusion of a \emph{single} item in business applications. Instead, we incorporate item
inference in spread maximization to estimate the influence more accurately.

\section{Social Item Graph Model}

\label{sec:social item graph}

Here we first present the social item graph model and then
introduce the diffusion process in the proposed model.

\subsection{Social Item Graph\label{sec:SIG}}

We aim to model user purchases and potential activations of new purchase actions from some prior. We first
define the notions of the social network and purchase actions.

\vspace{-5pt}
\begin{definition}
	A \textit{social network} is denoted by a directed graph $G=\left(
	V,E\right) $ where $V$ contains all the nodes and $E$ contains all the
	directed edges in the graph. Accordingly, a social network is also referred
	to as a \emph{social graph}.
\end{definition}

\vspace{-7pt}
\begin{definition}
	Given a list of commodity items $I$ and a set of customers $V$, a \textit{%
		purchase action}  (or \textit{purchase} for short), denoted by $(v,i)$ where $v\in V$ is a customer, and $i \in I$ is an item,
	refers to the purchase of item $i$ by customer $v$.
\end{definition}

\vspace{-7pt}
\begin{definition}
	An \textit{purchase log} is a database consisting of all the purchase actions
	in a given period of time.
\end{definition}
\vspace{-3pt}

\begin{figure}[t]
	\centering
	\includegraphics[width=1.8 in]{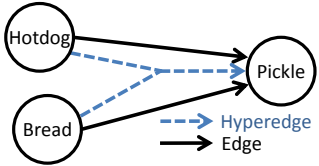} 
	\vspace{-10pt}
	\caption{A hyperedge example}
	\label{fig:hyperedge_example}
	\vspace{-2mm}
\end{figure}

\sloppy
Association-rule mining (called item inference in this paper) has been widely exploited to discover correlations between purchases in transactions. For example, the rule $\{\text{hotdog},\text{bread}\}\rightarrow \{\text{pickle}\}$ obtained from the transactions of a supermarket indicates that if a customer buys hotdogs and bread together, she is likely to buy pickles. To model the above likelihood, the confidence \cite{DataMiningTextbook11} of a rule $\{\text{hotdog},\text{bread}\}\rightarrow \{\text{pickle}\}$ is the proportion of the transactions that have hotdogs and bread also include pickles. It has been regarded as the conditional probability that a customer buying \textit{both} hotdogs \textit{and} bread would trigger the additional purchase of pickles. To model the above rule in a graph, a possible way is to use two separate  \textit{edges} (see Figure~\ref{fig:hyperedge_example}; one from hotdog to pickle, and the other from bread to pickle, respectively), while the probability associated with each of these edges is the confidence of the rule. In the above graph model, however, \textit{either one} of the hotdog or bread may trigger the purchase of pickle. This does not accurately express the intended condition of purchasing  \textit{both} the hotdog \textit{and} bread. By contrast, the \textit{hyperedges} in Graph Theory, by spanning multiple source nodes and one destination node, can model the above association rule (as illustrated in Figure \ref{fig:hyperedge_example}). The probability associated with the hyperedge represents the likelihood of the purchase action denoted by the destination node when \textit{all} purchase actions denoted by source nodes have happened.

On the other hand, in viral marketing, the traditional IC model activates a new node by the social influence probabilities associated with edges to the node. Aiming to capture both effects of item inference and social influence. We propose a new \textit{Social Item Graph} (SIG). SIG models the likelihood for a purchase (or a set of purchases) to trigger another purchase in form of \textit{hyperedges}, which may have one or multiple source nodes leading to one destination node. We define a social item graph as follows.


\begin{definition}
	Given a social graph of customers $G=(V, E)$ and a commodity item list $I$, a  \textit{social item graph} is denoted by $G_{SI}=\left(
	V_{SI},E_{H}\right) $, where $V_{SI}$ is a set of purchase actions
	and $E_{H}$ is a set of hyperedges over $V_{SI}$. A node $n \in V_SI$ is denoted as $(v,i)$, where $v\in V$ and $i\in I$.
	A hyperedge $e\in E_{H}$ is of the following form:
	\small
	\begin{eqnarray*}
		&&\{(u_{1},i_{1}),(u_{2},i_{2}),\cdots ,(u_{m},i_{m})\}\rightarrow (v,i)
	\end{eqnarray*}
	\noindent where $u_{i}$ is in the neighborhood of $v$ in $G$, i.e., $%
	u_{i}\in N_{G}\left( v\right) =\left\{ u|d(u,v)\leq 1\right\}$.\footnote{Notice that when $u_{1}=u_{2}=\cdots=u_{m}=v$, the hyperedge represents the item inference of item $i$. On the other hand, when $i_{1}=i_{2}=\cdots=i_{m}=i$, it becomes the social influence of user $u$ on $v$.}
\end{definition}
\vspace{-3pt}
Note that the conventional social influence edge in a social graph with one source and one destination can still be modeled in an SIG as a simple edge associated with a corresponding influence probability. Nevertheless, the influence probability from a person to another can vary for different items (e.g., a person's influence on another person for cosmetics and smartphones may vary.). Moreover, although an SIG may model the purchases more accurately with the help of \textit{both} social influence and item inference, the complexity of processing an SIG with hyperedges is much higher than simple edges in the traditional social graph that denotes only social influence.\footnote{To solve this issue, one approach is to transform a SIG with hyperedges to a graph without hyperedges, by replacing a hyperedge with multiple simple edges connecting to the sources and destinations, or by aggregating the source nodes and destination nodes into two nodes, respectively. Nevertheless, as to be shown in Section 4, the above strategies do not work. Also, a destination node can be activated only if all source nodes of the hyperedge are activated (see Section \ref{sec:diffSIG}).}

For simplicity, let $\mathtt{u}$ and $\mathtt{v}$ (i.e., the symbols in Typewriter style) represent the nodes $(u,i)$ and $(v,i)$ in SIG for the rest of this paper. We also denote a hyperedge as $e\equiv \mathtt{U}\rightarrow \mathtt{v}$, where $\mathtt{U}$ is a set of source nodes and $\mathtt{v}$ is the destination node. Let the associated edge weight be $p_{e}$, which represents the \emph{activation probability} for $\mathtt{v}$ to be activated if all source nodes in $\mathtt{U}$ are activated. Note that the activation probability is for one single hyperedge $\mathtt{U}\rightarrow \mathtt{v}$. Other hyperedges sharing the same destination may have different activation probabilities. For example, part of the source nodes in a hyperedge $\{\mathtt{a},\mathtt{b},\mathtt{c},\mathtt{d}\}\rightarrow \mathtt{x}$ can still activate $\mathtt{x}$, e.g., by $\{\mathtt{a},\mathtt{b},\mathtt{c}\}\rightarrow \mathtt{x}$ with a different hyperedge with its own activation probability.

\subsection{Diffusion Process in Social Item Graph}
\label{sec:diffSIG}
Next we introduce the diffusion process in SIG, which is inspired by the probability-based approach behind \emph{Independent Cascade} (IC) to captures the word-of-mouth behavior in the real world~\cite{Chen10KDD}.\footnote{Notice that diffusion process in SIG is based on IC model since it only requires one diffusion probability parameter associated to each edge whereas LT model requires both influence degree of each edge and an influence threshold for each node. Moreover, several variants of IC model have been proposed \cite{IC1,IC2}. However, they focus on modeling the diffusion process between users, such as aspect awareness \cite{IC1}, which is not suitable for social item graph since the topic is embedded in each SIG node.} This diffusion process fits the item inferences captured in an SIG naturally, as we can derive conditional probabilities on hyperedges to describe the trigger (activation) of purchase actions on a potential purchase.

The diffusion process in SIG starts with all nodes inactive initially.
Let $\texttt{S}$ denote a set of seeds (purchase actions). Let a node $\mathtt{s}\in \texttt{S}$ be a seed. It immediately becomes active. Given all the nodes in a source set $\mathtt{U}$ at iteration $\iota-1$, if they are all active at iteration $\iota$, a hyperedge $e\equiv \mathtt{U}\rightarrow \mathtt{v}$ has a chance to activate the inactive $\mathtt{v}$ with probability $p_{e} $. Each node $(v,i)$ can be activated once, but it can try to activate other nodes multiple times, one for each incident hyperedges. For the seed selection problem that we target on, the total number of activated nodes represents the number of items adopted by customers (called \textit{total adoptions} for the rest of this paper).

\section{Social Item Maximization}

\label{sec:social_item_maximization}

Upon the proposed Social Item Graph (SIG), we now formulate a new
seed selection problem, called \emph{Social Item Maximization Problem}
(SIMP), that selects a set of seed purchase actions to maximize potential sales or revenue in a marketing campaign.
In Section~\ref{sec:IdenHyperedge}, we will describe how to construct the SIG from purchase logs by a machine learning approach.


\begin{definition}
	Given a seed number $k$, a list of targeted items $I$, and a social item
	graph $G_{SI}(V_{SI},E_{H})$, SIMP selects a set $\texttt{S}$ of $k$ seeds in $V_{SI}$ such that $%
	\alpha _{G_{SI}}(\texttt{S})$, \textbf{the \emph{total adoption} function of $\texttt{S}$}, is
	maximized.
\end{definition}


Note that a seed in SIG represents the adoption/purchase action of a specific item by a
particular customer. The total adoption function $\alpha _{G_{SI}}$
represents the total number of product items ($\in I$) purchased. By assigning
prices to products and costs to the selected seeds, an extension of SIMP is to maximize the total revenue subtracted by the cost.

Here we first discuss the challenges in solving SIMP before introducing our algorithm. Note that, for the influence maximization
problem based on the IC model, Kempe et al. propose a $1-1/e$
approximation algorithm \cite{Kempe03KDD}, thanks to the submodularity in the problem. Unfortunately, the submodularity does not
hold for the total adoption function $\alpha _{G_{SI}}(\texttt{S})$ in SIMP.
Specifically, if the function $\alpha _{G_{SI}}(\texttt{S})$ satisfies the submodularity, for any node $\texttt{i}$ and any two subsets of nodes $\texttt{S}_1$ and $\texttt{S}_2$ where $\texttt{S}_1 \subseteq \texttt{S}_2$, $\alpha_{G_{SI}}(\texttt{S}_{1}\bigcup \{\texttt{i}\})-\alpha _{G_{SI}}(\texttt{S}_{1}) \geq \alpha_{G_{SI}}(\texttt{S}_{2}\bigcup \{\texttt{i}\})-\alpha _{G_{SI}}(\texttt{S}_{2})$ should hold. However, a counter example is illustrated below.

\begin{figure}[t]
	\begin{minipage}[t]{0.46\linewidth}
		\centering
		\vspace{7pt}
		\includegraphics[width=1.0in]{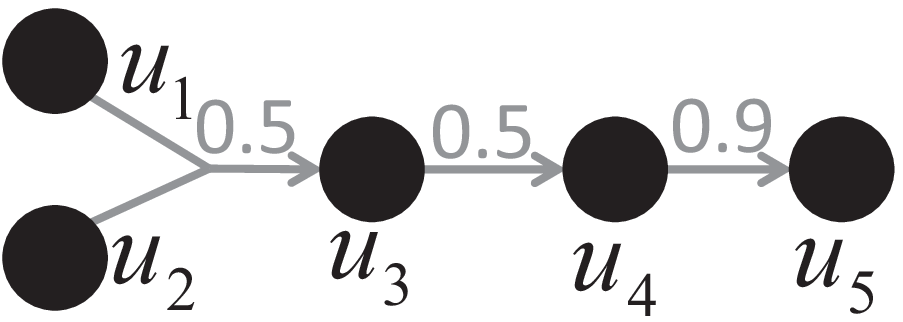}
		\caption{A non-submodular example}
		\label{fig:non_submodular}
		\vspace{-8pt}
	\end{minipage}
	\hspace{2pt}
	\begin{minipage}[t]{0.46\linewidth}
		\centering
		\vspace{0pt}
		\includegraphics[width=1.1in]{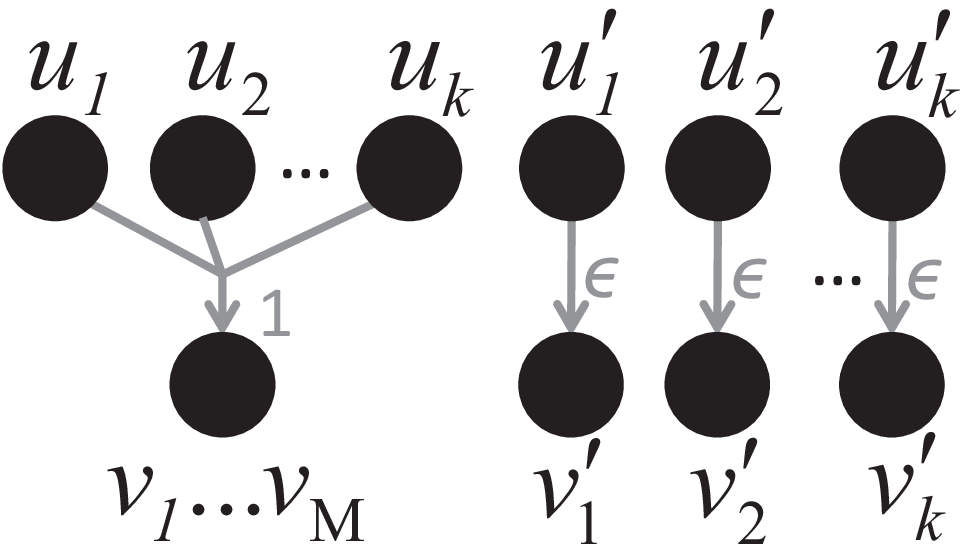}
		\vspace{-10pt}
		\caption{An example of SIMP}\label{fig:greedy_bad_case}
	\end{minipage}
\end{figure}

\begin{example}
	\label{example:non_submodular}
	\sloppy
	Consider an SIMP instance with a customer and
	five items in Figure~\ref{fig:non_submodular}. Consider the case where $%
	\texttt{S}_{1}=\{\texttt{u}_{4}\}$, $\texttt{S}_{2}=\{\texttt{u}_{1},\texttt{u}_{4}\}$, and $\texttt{i}$ corresponds to node $%
	\texttt{u}_{2}$. For seed sets $\{\texttt{u}_{4}\}$, $\{\texttt{u}_{2},\texttt{u}_{4}\}$, $\{\texttt{u}_{1},\texttt{u}_{4}\}$ and $%
	\{\texttt{u}_{1},\texttt{u}_{2},\texttt{u}_{4}\}$, $\alpha _{G_{SI}}(\{\texttt{u}_{4}\})=1.9$, $\alpha
	_{G_{SI}}(\{\texttt{u}_{2},\texttt{u}_{4}\})=2.9$, $\alpha _{G_{SI}}(\{\texttt{u}_{1},\texttt{u}_{4}\})=2.9$,
	and $\alpha _{G_{SI}}(\{\texttt{u}_{1},\texttt{u}_{2},\texttt{u}_{4}\})=4.4$. Thus, $\alpha
	_{G_{SI}}(\texttt{S}_{1}\bigcup \{\texttt{u}_{2}\})-\alpha _{G_{SI}}(\texttt{S}_{1})=1<1.5=\alpha _{G_{SI}}(\texttt{S}_{2}\bigcup \{\texttt{u}_{2}\})-\alpha
	_{G_{SI}}(\texttt{S}_{2})$. Hence, the submodularity does not hold.
\end{example}

Since the submodularity does not exist in SIMP, the $1-1/e$
approximation ratio of the greedy algorithm in \cite{Kempe03KDD} does not
hold. Now, an interesting question is how large the ratio becomes.
Example~\ref{example:greedy_bad_case} shows an SIMP instance where the
greedy algorithm performs poorly.

\begin{example}
	Consider an example in Figure~\ref{fig:greedy_bad_case}, where nodes $\texttt{v}_{1}$%
	, $\texttt{v}_{2}$,...,$\texttt{v}_{M}$ all have a hyperedge with the probability as 1 from
	the same $k$ sources $\texttt{u}_{1}$, $\texttt{u}_{2}$,..., $\texttt{u}_{k}$, and $\epsilon$ is an
	arbitrarily small edge probability $\epsilon >0$. The greedy algorithm
	selects one node in each iteration, i.e., it selects $\texttt{u}_{1}^{\prime }$, $%
	\texttt{u}_{2}^{\prime }$...$\texttt{u}_{k}^{\prime }$ as the seeds with a total adoption $%
	k+k\epsilon $. However, the optimal solution actually selects $\texttt{u}_{1}$, $\texttt{u}_{2}$%
	,..., $\texttt{u}_{k}$ as the seeds and results in the total adoption $M+k$.
	Therefore, the approximation ratio of the greedy algorithm is at least $%
	(M+k)/(k+k\epsilon )$, which is close to $M/k$ for a large $M$, where $M$
	could approach $\left\vert V_{SI}\right\vert $ in the worst case. \label%
	{example:greedy_bad_case}
\end{example}

\begin{figure}[t]
	\centering
	\subfigure[Original] {\includegraphics[height =
		0.6 in]{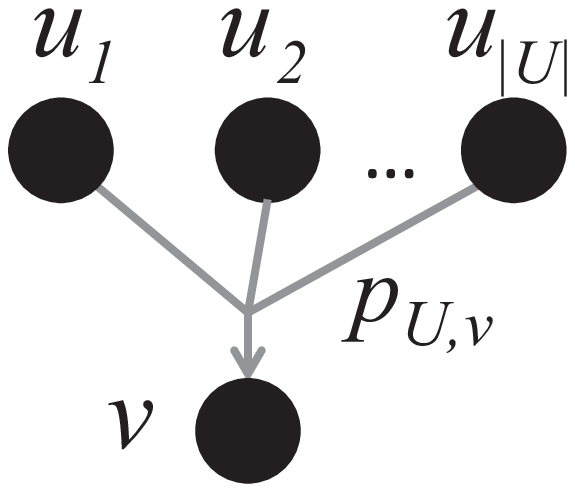}} \hspace{0.8 in}
	\subfigure[Transformed] {\includegraphics[height =
		0.6 in]{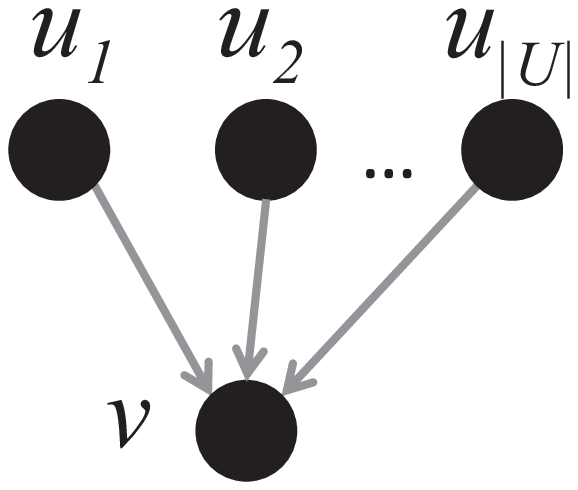}} \vspace{-10pt}
	\caption{An illustration of graph transformations}
	\label{fig:transformation_illustration}
\end{figure}

One may argue that the above challenges in SIMP may be alleviated
by transforming $G_{SI}$ into a graph with only simple edges,
as displayed in Figure~\ref{fig:transformation_illustration}, where the
weight of every $\texttt{u}_{i} \rightarrow \texttt{v})$ with $\texttt{u}_{i}\in \texttt{U}$ can be set independently.
However, if a source node $\texttt{u}_{m}\in \texttt{U}$ of $\texttt{v}$ is difficult to activate,
the probability for $\texttt{v}$ to be activated approaches zero in Figure~\ref%
{fig:transformation_illustration}~(a) due to $\texttt{u}_{m}$. However, in
Figure~\ref{fig:transformation_illustration}~(b), the destination  $\texttt{v}$
is inclined to be activated by sources in $\texttt{U}$, especially when $\texttt{U}$
is sufficiently large. Thus, the idea of graph transformation does not work.

\subsection{Hyperedge-Aware Greedy (HAG)}

\label{subsec:SIMP_algo}
Here, we propose an  algorithm for SIMP, \emph{Hyperedge-Aware Greedy (HAG)}, with performance guarantee. The approximation ratio is proved in Section~\ref{subsec:hardness}.
A hyperedge requires all its sources activated first in order to
activate  the destination. Conventional single node greedy algorithms
perform poorly because hyperedges are not considered. To address this important issue, we propose \textit{%
	Hyperedge-Aware Greedy (HAG)} to select multiple seeds in each
iteration.

A naive algorithm for SIMP would examine, $C_{k}^{|V_{SI}|}$ combinations are involved to choose $k$ seeds.
In this paper, as multiple seeds tend to activate all source nodes
of a hyperedge in order to activate its destination, an effective way is to consider only the
combinations which include the source nodes of any hyperedge. We call the
source nodes of a hyperedge as a \emph{source combination}. Based on this
idea, in each iteration, HAG includes the source combination leading to
\emph{the largest increment on total adoption divided by the number of new
	seeds added in this iteration}. Note that only the source combinations with
no more than $k$ sources are considered. The iteration continues until $k$
seeds are selected. Note that HAG does not restrict the seeds to be the
source nodes of hyperedges. Instead, the source node $\mathtt{u}$ of any
simple edge $\mathtt{u} \rightarrow \mathtt{v}$ in SIG is also examined.


\noindent \textbf{Complexity of HAG. }To select $k$ seeds, HAG takes at most
$k$ rounds. In each round, the source combinations of $|E_{H}|$ hyperedges
are tried one by one, and the diffusion cost is $c_{dif}$, which will be
analyzed in Section~\ref{subsec:mc_acceleration}. Thus, the time complexity
of HAG is $O(k\times |E_{H}|\times c_{dif})$.

\subsection{Acceleration of Diffusion Computation}

\label{subsec:mc_acceleration}

To estimate the total adoption for a seed set, it is necessary to perform Monte Carlo simulation based on the
diffusion process described in Section~\ref{sec:diffSIG} for many times.
Finding the total adoption is very expensive, especially when a node $\texttt{v}$ can be activated by a hyperedge with a large source set $\texttt{U}$, which indicates that there also exist many other hyperedges with an arbitrary subset of $\texttt{U}$ as the source set to activate $\texttt{v}$.
In other words, enormous hyperedges need to be examined for the diffusion on an SIG. It is essential to reduce the computational overhead.
To address this issue, we propose a new index structure, called \emph{SIG-index}, by exploiting FP-Tree~\cite{DataMiningTextbook11}
to pre-process source combinations in
hyperedges in a compact form in order to facilitate efficient derivation of activation probabilities during the diffusion process.

The basic idea behind SIG-index is as follows.
For each node $\mathtt{v}$ with the set of activated
in-neighbors $N_{\mathtt{v},\iota}^{a}$ in iteration $\iota$, if $\mathtt{v}$
has not been activated before $\iota$, the diffusion process will try to activate
$\mathtt{v}$ via every hyperedge $\mathtt{U}\rightarrow \mathtt{v}$
where the last source in $\mathtt{U}$ has been activated in iteration $\iota-1$.
To derive the activation probability of a node $\mathtt{v}$ from the weights of hyperedges associated with $%
\mathtt{v}$, we first define the activation probability as follows.

\begin{definition}
	The activation probability of $\mathtt{v}$ at $\iota$ is
	\begin{equation*}
	ap_{\mathtt{v},\iota} = 1 - \prod_{\mathtt{U} \rightarrow \mathtt{v}\in E_H,
		\mathtt{U} \subseteq N_{\mathtt{v},\iota-1}, \mathtt{U} \nsubseteq N_{\mathtt{v}%
			,\iota-2}} {(1-p_{\mathtt{U} \rightarrow \mathtt{v}})}.
	\end{equation*}
	where $N_{\mathtt{v},\iota-1}$ and $N_{\mathtt{v},\iota-2}$ denote the set of
	active neighbors of $\mathtt{v}$ in iteration $\iota-1$ and $\iota -2$,
	respectively. \label{def:ap}
\end{definition}

The operations on an SIG-index occur two phases: Index Creation Phase and Diffusion Processing Phase. As all hyperedges satisfying Definition~\ref{def:ap} must be
accessed, the SIG-index stores the hyperedge probabilities in Index Creation Phase. Later, the SIG-index is updated in Diffusion Processing Phase to derive the activation probability efficiently.

\textbf{Index Creation Phase.} For each hyperedge $\mathtt{U}\rightarrow \mathtt{v}
$, we first regard each source combination $\mathtt{\mathtt{U}}=\{\mathtt{v}%
_{1},...\mathtt{v}_{|\mathtt{U}|}\}$ as a transaction to build an FP-tree~\cite{DataMiningTextbook11}
by setting the minimum support as 1. As such,  $\mathtt{v}_{1},...\mathtt{v}_{|\mathtt{%
		U}|}$ forms a path $r\rightarrow \mathtt{v}_{1}\rightarrow \mathtt{v}%
_{2}...\rightarrow \mathtt{v}_{|\mathtt{U}|}$ from the root $r$
in the FP-tree to node $\mathtt{v}_{|\mathtt{U}|}$ in $\mathtt{\mathtt{U}}$.
Different from the FP-Tree, the SIG-index associates the
probability of each hyperedge $\mathtt{U}\rightarrow \mathtt{v}$ with the
last source node $\mathtt{v}_{|\mathtt{U}|}$ in $\mathtt{\mathtt{U}}$.%
\footnote{%
	For ease of explanation, we assume the order of nodes in the SIG-index
	follows the ascending order of subscript.}
Initially the probability associated with the root $r$ is 0. Later the
the SIG-index is updated during the diffusion process.
Example~\ref{eg_3} illustrates the SIG-index created based on an SIG.

\begin{figure}[t]
	\centering
	\subfigure[Initial] {\includegraphics[width =
		1.5 in]{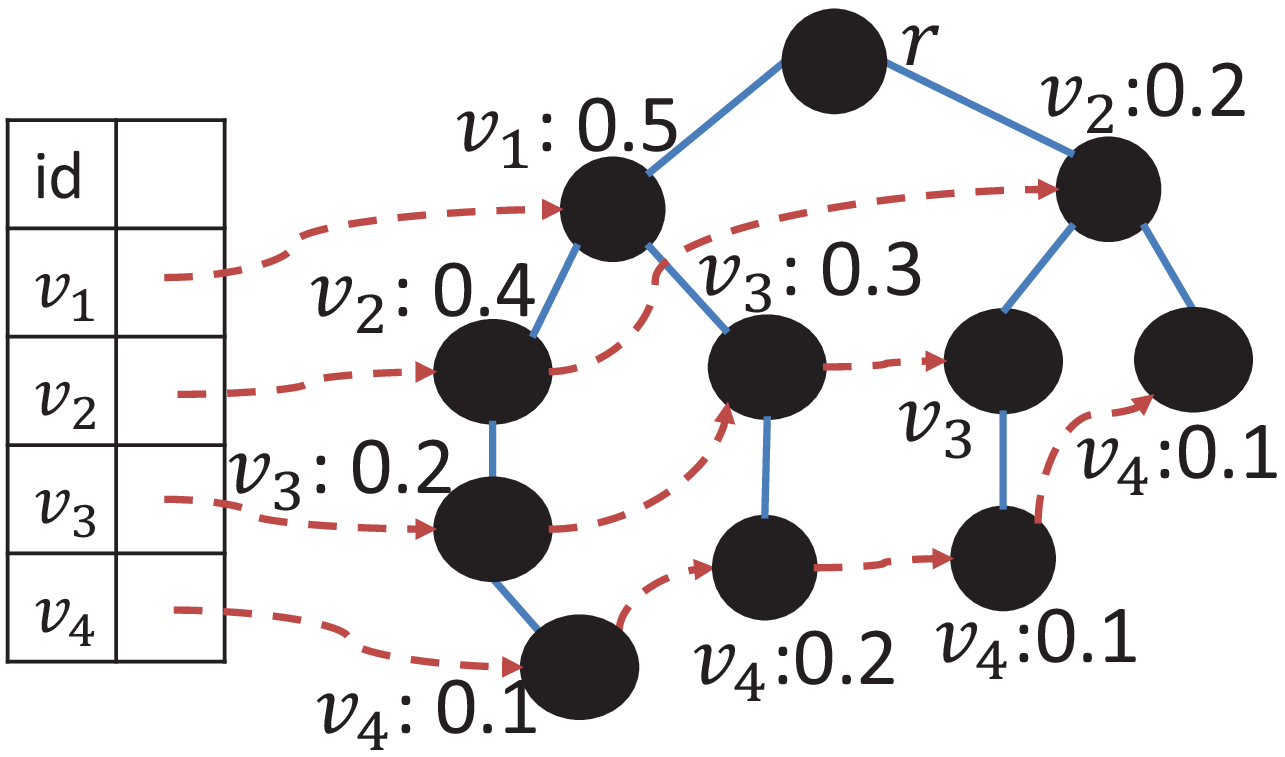}} \vspace{-10pt}
	\subfigure[After $\texttt{v}_2$ is activated] {\includegraphics[width = 1.5
		in]{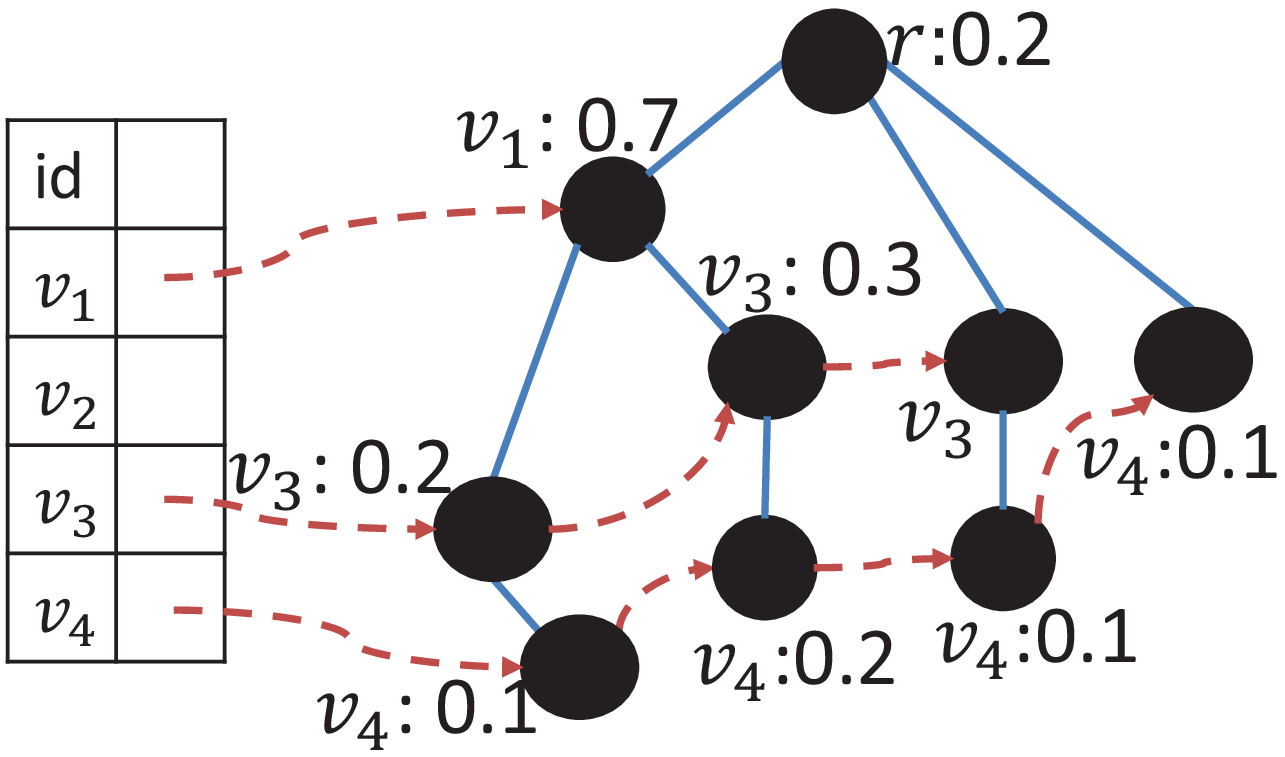}} 
	\caption{An illustration of SIG-index}
	\label{fig:fp_index}
\end{figure}

\begin{example}
\label{eg_3}
	Consider an SIG graph with five nodes, $\mathtt{v}_{1}$-$\mathtt{v}_{5}$, and nine hyperedges with their associated probabilities in parentheses:\footnote{For simplicity, the hyperedges in this example only have $\mathtt{v}_5$ as the destination.} $\{%
	\mathtt{v}_{1}\}\rightarrow \mathtt{v}_{5}$ (0.5), $\{\mathtt{v}_{1},\mathtt{%
		v}_{2}\}\rightarrow \mathtt{v}_{5}$ (0.4), $\{\mathtt{v}_{1},\mathtt{v}_{2},%
	\mathtt{v}_{3}\}\rightarrow \mathtt{v}_{5}$ (0.2), $\{\mathtt{v}_{1},\mathtt{%
		v}_{2},\mathtt{v}_{3},\mathtt{v}_{4}\}\rightarrow \mathtt{v}_{5}$ (0.1), $\{%
	\mathtt{v}_{1},\mathtt{v}_{3}\}\rightarrow \mathtt{v}_{5}$ (0.3), $\{\mathtt{%
		v}_{1},\mathtt{v}_{3},\mathtt{v}_{4}\}\rightarrow \mathtt{v}_{5}$ (0.2), $\{%
	\mathtt{v}_{2}\}\rightarrow \mathtt{v}_{5}$ (0.2), $\{\mathtt{v}_{2},\mathtt{%
		v}_{3},\mathtt{v}_{4}\}\rightarrow \mathtt{v}_{5}$ (0.1), $\{\mathtt{v}_{2},%
	\mathtt{v}_{4}\}\rightarrow \mathtt{v}_{5}$ (0.1). Figure~\ref{fig:fp_index}
	(a) shows the SIG-index initially created for node $\mathtt{v}_{5}$.
\end{example}

\textbf{Diffusion Processing Phase.} The activation probability in an iteration is
derived by traversing the initial SIG-index, which takes $O(|E_{H}|)$ time.
However, a simulation may iterate a lot of times. To further accelerate the
traversing process, we adjust the SIG-index for the
activated nodes in each iteration. More specifically, after a node $\mathtt{v}%
^{a}$ is activated, accessing an hyperedge $\mathtt{U}\rightarrow \mathtt{v}$
with $\mathtt{v}^{a}\in \mathtt{U}$ becomes easier since the number
remaining inactivated nodes in $\mathtt{U}-\{\mathtt{v}^{a}\}$ is reduced.
Accordingly, SIG-index is modified by traversing every vertex labeled as $\mathtt{v}%
^{a}$ on the SIG-index in the following steps.
1) If $\mathtt{v}^{a}$ is associated with a
probability $p_{a}$, it is crucial to aggregate the old activation
probabilities $p_{a}$ of $\mathtt{v}^{a}$ and $p_{p}$ of its parent $\mathtt{v}^{p}$, and update activation probability
associated with $\mathtt{v}^{p}$ as $1-(1-p_{a})(1-p_{p})$, since the source combination needed for accessing the hyperedges
associated with $\mathtt{v}^{a}$ and $\mathtt{v}^{p}$
becomes the same. The aggregation is also performed when $\mathtt{v}^{p}$ is $r$.
2) If $\mathtt{v}^{a}$ has any
children $\mathtt{c}$, the parent of $\mathtt{c}$ is
changed to be
$\mathtt{v}^{p}$, which removes the processed $\mathtt{v}^{a}$ from the index.
3) After processing every node $\mathtt{v}^{a}$ in the SIG-index, we obtain  the activation probability of $\mathtt{v}$ in the root $r$. After the probability is accessed for activating $\texttt{v}$, the probability of $r$ is reset to 0 for next iteration.

\begin{example}
	Consider an example with $\mathtt{v}_{2}$ activated in an iteration. To
	update the SIG-index, each vertex $\mathtt{v}_{2}$ in Figure~\ref%
	{fig:fp_index}~(a) is examined  by traversing the linked list of $\mathtt{v}_{2}$. First, the left vertex with label
	$\mathtt{v}_{2}$ is examined. SIG-index reassigns the parent of $\mathtt{v}%
	_{2}$'s child (labeled as $\mathtt{v}_{3}$) to the vertex labeled as $%
	\mathtt{v}_{1}$, and aggregate the probability 0.4 on the $\mathtt{v}_{2}$ and 0.5 on vertex $\mathtt{v}_{1}$, since the hyperedge $\{\mathtt{v%
	}_{1},\mathtt{v}_{2}\}\rightarrow \mathtt{v}_{5}$ can be accessed if the
	node $\mathtt{v}_{1}$ is activated later. The probability of
	$v_{1}$ becomes $1-(1-p_{\mathtt{v}_{1}})(1-p_{\mathtt{v}_{2}})=0.7$. Then
	the right vertex with label $\mathtt{v}_{2}$ is examined. The parent of
	its two children is reassigned to the root $r$. Also, the
	probability of itself (0.2) is aggregated with the root $r$,
	indicating that the activation probability of node $v_{5}$ in the next
	iteration is 0.2.
\end{example}

\noindent \textbf{Complexity Analysis. } For Index Creation Phase, the
initial SIG-index for $\mathtt{v}$ is built by examining the hyperedges two times with
$O(|E_{H}|)$ time. The number of vertices in SIG-index is at most $%
O(c_{d}|E_{H}|)$, where $c_{d}$ is the number of source nodes in the
largest hyperedge. During Diffusion Processing Phase, each vertex in SIG-index is examined
only once through the node-links, and the parent of a vertex is changed at
most $O(c_{d})$ times. Thus, the overall time to complete a diffusion
requires at most $O(c_{d}|E_{H}|)$ time.

\subsection{Hardness Results}
\label{subsec:hardness}

\begin{figure}[t]
	\centering
	\includegraphics[width=2 in]{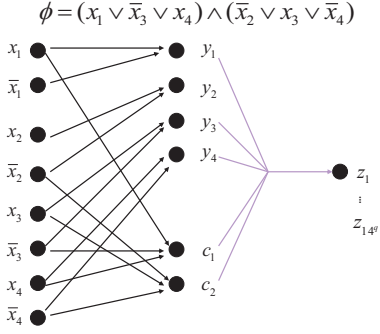}
	\vspace{-10pt}
	\caption{An illustration instance built for 3-SAT}
	\label{fig:reduction_3SAT}
\end{figure}

From the discussion earlier, it becomes obvious that SIMP is difficult. In the following, we will prove that
SIMP is inapproximable with a non-constant ratio $n^{c}$ for all $c<1$, with a gap-introducing reduction from an NP-complete problem 3-SAT to SIMP, where $n$ is the number of nodes in an SIG.

Given an expression $\phi$ in CNF, in which each clause has three variables, 3-SAT is to decide whether $\phi$ is satisfiable. The reduction includes two parts: 1) If $\phi$ is satisfiable,  its transformed SIMP instance has optimal total adoption larger than $\alpha_{SAT}$. 2) If $\phi$ is unsatisfiable, its transformed SIMP instance has optimal total adoption less than $\alpha_{UNSAT}$. (Refer Lemma~\ref{keylemma} for actual $\alpha_{SAT}$ and $\alpha_{UNSAT}$.)
Then inapproximability obtained by this gap-introducing induction is $\frac{\alpha_{SAT}}{\alpha_{UNSAT}}$. Note that an $\frac{\alpha_{SAT}}{\alpha_{UNSAT}}$-approximation algorithm is able to solve 3-SAT because it always returns a solution larger than $\frac{\alpha_{UNSAT}}{\alpha_{SAT}} \times \alpha_{SAT} = \alpha_{UNSAT}$ for satisfiable $\phi$ and a solution smaller than $\alpha_{UNSAT}$ for unsatisfiable $\phi$, implying such approximation algorithm must not exist. Also note that the theoretical result only shows that for any algorithm, there exists a problem instance of SIMP (i.e., a pair of an SIG graph and a seed number $k$) that the algorithm can not obtain a solution better than $1/n$ times the optimal solution. It does not imply that an algorithm always performs badly in every SIMP instance.




\begin{lemma}
	\label{keylemma} For a positive integer $q$, there is a
	gap-introducing reduction from 3-SAT to SIMP, which transforms an $n_{var}$%
	-variables expression $\phi$ to an SIMP instance with the SIG as $G_{SI}(V_{SI},E_H)$ and the $k$ as $%
	n_{var}$ such that \newline
	$\bullet $ if $\phi $ is satisfiable, $\alpha_{G_{SI}}^*\geq
	(m_{cla}+3n_{var})^{q}$, and \newline
	$\bullet $ if $\phi $ is not satisfiable, $\alpha_{G_{SI}}^* < m_{cla}+3n_{var}$, \newline
	where $\alpha_{G_{SI}}^*$ is the optimal solution of this instance, $n_{var}$ is the number of Boolean variables, and $m_{cla}$ is the number
	of clauses.
	Hence there is no $(m_{cla}+3n_{var})^{q-1}$ approximation
	algorithm for SIMP unless P $= $ NP.
\end{lemma}
\begin{proof}
	Given a positive integer $q$, for an instance $%
	\phi $ of 3-SAT with $n_{var}$ Boolean variables $a_{1},\dots ,a_{n_{var}}$
	and $m_{cla}$ clauses $C_{1},\dots, C_{m_{cla}}$, we construct an SIG $%
	G_{SI} $ with three node sets $\texttt{X}$, $\texttt{Y}$ and $\texttt{Z}$ as follows. 1) Each Boolean
	variable $a_{i}$ corresponds to two nodes $\texttt{x}_{i}$, $\overline{\texttt{x}}_{i}$ in $\texttt{X}$
	and one node $\texttt{y}_{i}$ in $\texttt{Y}$. 2) Each clause $C_{k}$ corresponds to one node
	$\texttt{c}_{k}$ in $\texttt{Y}$. 3) $\texttt{Z}$ has
	$(|\texttt{X}|+|\texttt{Y}|)^{q}$ nodes. (Thus, $G_{SI}$ has
	$(m_{cla}+3n_{var})^{q}+m_{cla}+3n_{var}$ nodes.) 4) For each $\texttt{y}_{j}$ in $\texttt{Y}$, we	
	add direct edges $\texttt{x}_{j} \rightarrow \texttt{y}_{j}$ and $\overline{\texttt{x}}_{j} \rightarrow \texttt{y}_{j}$. 5) For each
	$\texttt{c}_{k}$ in $\texttt{Y}$, we add direct edges $\alpha \rightarrow \texttt{c}_{k}$, $\beta \rightarrow \texttt{c}_{k}$ and
	$\gamma \rightarrow \texttt{c}_{k}$, where $\alpha$, $\beta$, $\gamma$ are the nodes in $\texttt{X}$
	corresponding to the three literals in $C_{k}$. 6) We add a hyperedge $%
	\texttt{Y} \rightarrow \texttt{z}_v$ from all for every $\texttt{z}_v\in \texttt{Z}$. The probability of every
	edge is set to 1.  An example is illustrated in Figure~\ref{fig:reduction_3SAT}.
	
	We first prove that $\phi$ is satisfiable if and only if $G_{SI}$ has a
	seed set $\texttt{S}$ with $n_{var}$ seeds and the total adoption of $\texttt{S}$ contains $\texttt{Y}$. If $\phi$ is satisfiable, there exists a truth assignment $T$ on Boolean variables $a_{1},\dots ,a_{n_{var}}$ satisfying all clauses of $\phi$. Let $\texttt{S}=\{\texttt{x}_{i}|T(a_{i})=1\}\cup \{\overline{
		\texttt{x}}_{j}|T(a_{j})=0\}$, and $\texttt{S}$ then has $n_{var}$ nodes and the total adoption of $\texttt{S}$ contains $\texttt{Y}$. On the other hand, if $\phi$ is not satisfiable, apparently there
	exists no seed set $\texttt{S}$  with exactly one of $\texttt{x}_{i}$ or $\overline{\texttt{x}}_{i}$ selected for every $i$ such that the total adoption of $\texttt{S}$ contains $\texttt{Y}$. For other cases, 1) all seeds are placed in $\texttt{X}$, but there exists at least one $i$ with both $\texttt{x}_{i}$ and $\overline{\texttt{x}}_{i}$
	selected. In this case, there must exist some $j$ such that none
	of $\texttt{x}_{j}$ or $\overline{\texttt{x}}_{j}$ are selcted (since the
	seed number is $n_{var}$), and thus $\texttt{Y}$ is not covered by
	the total adoption of $\texttt{S}$. 2) A seed is placed in $\texttt{Y}$. In
	this case, the seed can be moved to an adjacent $\texttt{x}_{i}$ without reducing the total adoption. Nevertheless, as explained above,
	there exists no seed set $\texttt{S}$ with all seeds placed in $\texttt{X}$
	such that the total adoption of $\texttt{S}$ contains $\texttt{Y}$, and
	thus the total adoption of any seed set with a seed placed in $\texttt{Y}$
	cannot cover $\texttt{Y}$, either. With above observations, if $\phi$ is not satisfiable, $G_{SI}$ does not have a seed set $\texttt{S}$ with $n_{var}$ seeds such that the total adoption of $\texttt{S}$ contains $\texttt{Y}$.	
	Since the nodes of $\texttt{Z}$ can be activated if and only if the total adoption of
	$\texttt{S}$ contains $\texttt{Y}$ if and only if $\phi$ is satisfiable, we have \newline
	$\bullet$ if $\phi$ is satisfiable, $\alpha_{G_{SI}}^* \geq
	(m_{cla}+3n_{var})^q$, and \newline
	$\bullet$ if $\phi$ is not satisfiable, $\alpha_{G_{SI}}^* <
	m_{cla}+3n_{var}$. \newline
	The lemma follows.
\end{proof}

\begin{theorem}
	\label{hardness} For any $\epsilon >0$, there is no $n^{1-\epsilon}$
	approximation algorithm for SIMP, assuming P $\neq $ NP.
\end{theorem}
\begin{proof}
	For any arbitrary $\epsilon >0$, we set $q \geq \frac{2}{\epsilon }$. Then, by
	Lemma \ref{keylemma}, there is no $(m_{cla}+3n_{var})^{q-1}$ approximation
	algorithm for SIMP unless P $=$ NP. Then $(m_{cla}+3n_{var})^{q-1}\geq
	2(m_{cla}+3n_{var})^{q-2} \geq 2(m_{cla}+3n_{var})^{q(1-\epsilon )}\geq
	(2(m_{cla}+3n_{var})^{q})^{1-\epsilon }\geq n^{1-\epsilon }$. Since $%
	\epsilon $ is arbitrarily small, thus for any $\epsilon >0$, there is no $%
	n^{1-\epsilon }$ approximation algorithm for SIMP, assuming P $\neq $ NP.
	The theorem follows.
\end{proof}

With Theorem~\ref{hardness}, no algorithm can achieve an approximation ratio better than $n$.
In Theorem~\ref{thm:correctness}, we prove that SIG-index is correct, and HAG with SIG-index
achieves the best ratio, i.e., it is $n$-approximated to SIMP.
Note that the approximation ratio only guarantees the lower bound of total adoption obtained by HAG theoretically.
Later in Section~\ref{subsec:effectiveness}, we empirically show that the total adoption obtained by HAG is comparable to the optimal solution.

\begin{theorem}
	HAG with SIG-index is $n$-approximated, where $n$ is the number of nodes in SIG. \label{thm:correctness}
\end{theorem}
\begin{proof}
	First, we prove that SIG-index obtains $ap_{\texttt{v},\iota}$ correctly.
	Assume that there exists an incorrect $ap_{\texttt{v},\iota}$, i.e., there exists an hyperedge $\texttt{U} \rightarrow \texttt{v}$ satisfying the conditions in Definition~\ref{def:ap} (i.e., $\texttt{U} \nsubseteq N_{\texttt{v},\iota-2}$ and $\texttt{U} \subseteq N_{\texttt{v},\iota-1}$) but its probability is not aggregated to $r$ in $\iota$. However, the probability can not be aggregated before $\iota$ since $\texttt{U} \nsubseteq N_{v,\iota-2}$ and it must be aggregated no later than $\iota$ since $\texttt{U} \subseteq N_{\texttt{v},\iota-1}$. There is a contradiction.
	
	Proving that HAG with SIG-index is an $n$-approximation algorithm is
	simple. The upper bound of total adoption for the optimal algorithm is $n$,
	while the lower bound of the total adoption for HAG is $1$ because at least
	one seed is selected. In other words, designing an approximation algorithm
	for SIMP is simple, but it is much more difficult to have the hardness
	result for SIMP, and we have proven that SIMP is inapproximable within $%
	n^{1-\epsilon}$ for any arbitrarily small $\epsilon $.
\end{proof}

\begin{theorem}
	\label{hardness} For any $\epsilon >0$, there is no $n^{1-\epsilon }$
	approximation algorithm for SIMP, assuming P $\neq $ NP.
\end{theorem}
\begin{proof}
For any arbitrary $\epsilon >0$, we set $q \geq \frac{2}{\epsilon }$. Then, by
Lemma \ref{keylemma}, there is no $(m_{cla}+3n_{var})^{q-1}$ approximation
algorithm for SIMP unless P $=$ NP. Then $(m_{cla}+3n_{var})^{q-1}\geq
2(m_{cla}+3n_{var})^{q-2} \geq 2(m_{cla}+3n_{var})^{q(1-\epsilon )}\geq
(2(m_{cla}+3n_{var})^{q})^{1-\epsilon }\geq n^{1-\epsilon }$. Since $%
\epsilon $ is arbitrarily small, thus for any $\epsilon >0$, there is no $%
n^{1-\epsilon }$ approximation algorithm for SIMP, assuming P $\neq $ NP.
The theorem follows.
\end{proof}

\begin{corollary}
	HAG with SIG-index is $n$-approximated, where $n$ is the number of nodes in SIG. \label{thm:correctness}
\end{corollary}

Proving that HAG with SIG-index is an $n$-approximation algorithm is
simple. The upper bound of total adoption for the optimal algorithm is $n$,
while the lower bound of the total adoption for HAG is $1$ because at least
one seed is selected. Note that SIG-index is only for acceleration and does not change the solution quality.

\section{Construction of SIG }
\label{sec:IdenHyperedge}
To select seeds for SIMP, we need to construct the SIG from purchase logs and the social network.
We first create possible hyperedges by scanning the purchase logs.
Let $\tau$ be the timestamp of a given purchase $\texttt{v} = (v, i)$. $v$'s friends purchase and her own purchases that have happened within a given period before $\tau$ are considered as candidate source nodes to generate hyperedges
to $\texttt{v}$.\footnote{The considered periods of item inference and social influence can be different since social influence usually requires a longer time to propagate the messages while the item inference on the e-commerce websites can happen at the same time. The detail setting of time period will be discussed in Section \ref{expsetting}.} For each hyperedge $e$, the main task is then the estimation of its their activation probability $p_e$.
Since $p_{e}$ is unknown, it is estimated by maximizing the likelihood function based on observations in the purchase logs.  Note that learning the activation probability $p_e$ for each hyperedge $e$ faces three challenges.

\noindent \textbf{C1. Unknown distribution of $p_{e}$.} How to properly model $p_{e}$ is critical.\newline
\noindent \textbf{C2. Unobserved activations.} When $\mathtt{v}$ is activated at time $\tau$, this event only implies that at least one hyperedge successfully activates $\mathtt{v}$ before $\tau$. It remains unknown which or hyperedge(s) actually triggers $\texttt{v}$, i.e., it may be caused by either the item inference or social influence or both. Therefore, we cannot simply employ the confidence of an association-rule as the corresponding hyperedge probability.\newline
\noindent \textbf{C3. Data Sparsity.} The number of activations for a user to buy an item is small, whereas the number of possible hyperedge combinations is large. Moreover, new items emerge every day in e-commerce websites, which incurs the notorious cold-start problem. Hence, a method to deal with the data sparsity issue is necessary to properly model a SIG.

To address these challenges, we exploit a statistical inference approach to identify those hyperedges and learn their weights. In the following, we first propose a model of the edge function (to address the first challenge) and then exploit the smoothed expectation and maximization (EMS) algorithm~\cite{Silverman90}
to address the second and third challenges.

\subsection{Modeling of Hyperedge Probability}
\label{sec: ModelBuiling}

To overcome the first challenge, one possible way is to model the number of success activations and the number of unsuccessful activations by the binomial distributions. As such, $p_{e}$ is approximated by the ratio of the number of success activations and the number of total activation trials. However, the binomial distribution function is too complex for computing the maximum likelihood of a vast number of data. To handle big data, previous study reported \cite{Hogg05} that the binomial distribution $(n,p)$ can be approximated by the Poisson distribution $\lambda =np$ when the time duration is sufficiently large.
According to the above study, it is assumed that the number of activations of a hyperedge $e$ follows the Poisson distribution to handle the social influence and item inference jointly. The expected number of events equals to the intensity parameter $\lambda$. Moreover, we use an inhomogeneous Poisson process defined on the space of hyperedges to ensure that $p_{e}$ varies with different $e$.

In the following, a hyperedge is of size $n$, if the cardinality of its
source set $\mathtt{U}$ is $n$. We denote the intensity of the number of activation trials of the hyperedge $e$ as $\lambda _{T}(e)$. Then the successful activations of hyperedge $e$ follows another Poisson process where the intensity is denoted by $\lambda_{A}(e)$. Therefore, the hyperedge probability $p_{e}$ can be derived by parameters $\lambda _{A}(e)$ and $\lambda _{T}(e)$, i.e., $p_{e}=\frac{\lambda _{A}(e)}{\lambda _{T}(e)}$.

The maximum likelihood estimation can be employed to derive $\lambda_{T}(e)$. Nevertheless, $\lambda _{A}(e)$ cannot be derived as explained in the second challenge. Therefore, we use the expectation maximization (EM) algorithm, which is an extension of maximum likelihood estimation containing latent variables to $\lambda _{A}(e)$ which is modeled as the latent variable. Based on the observed purchase logs, the E-step first derives the likelihood Q-function of parameter $p_{e}$ with $\lambda _{A}(e)$ as the latent variables. In this step, the purchase logs and $p_{e}$ are given to find the probability function describing that all events on $e$ in the logs occur according to $p_{e}$, whereas the probability function (i.e., Q-function) explores different possible values on latent variable $\lambda _{A}(e)$. Afterward, The M-step maximizes the Q-function and derives the new $p_{e}$ for E-Step in the next iteration. These two steps are iterated until convergence.

With the employed Poisson distribution and EM algorithm, data sparsity remains
an issue. Therefore, we further exploit a variant of EM algorithm, called \emph{EMS algorithm} \cite{Silverman90}, to alleviate the sparsity problem by estimating the intensity of Poisson process using \textit{similar hyperedges}. The parameter smoothing after each iteration is called S-Step, which is incorporated in EMS algorithm, in addition to the existing E-Step and M-Step.

\subsection{Model Learning by EMS Algorithm}

\label{sec: EMS} Let $p_{e}$ and $\hat{p}_{e}$ denote the true probability and estimated probabilities for hyperedge $e$ in the EMS algorithm, respectively,
where $e=U\rightarrow v$. Let $N_{U}$ and $K_{e}$ denote the number of activations of source set $U$ in the purchase logs
and the number of successful activations on hyperedge $e$, respectively. The EM algorithm is exploited to find the maximum likelihood of $p_{e}$, while ${\lambda _{A}(e)}$ is the latent variable because $K_{e}$ cannot be observed (i.e., only $N_{U}$ can be observed). Therefore, E-Step derives the likelihood function for $\{p_{e}\}$ (i.e., the Q-function) as follows,
\vspace{-2pt}
\begin{equation}
Q(p_{e},\hat{p}_{e}^{(i-1)})=E_{K_{e}}[\log P(K_{e},N_{U}|p_{e})|N_{U},\hat{p}_{e}^{(i-1)}],
\vspace{-2pt}
\end{equation}
\noindent where $\hat{p}_{e}^{(i-1)}$ is the hyperedge probability derived in the previous iteration, Note that $N_{U}$ and $p_{e}^{(i-1)}$ are given parameters in iteration ${i}$, whereas $p_{e}$ is a variable in the Q-function, and $K_{e}$ is a random variable governed by the distribution $P(K_{e}|N_{U},p_{e}^{(i-1)})$. Since $p_{e}$ is correlated to $\lambda _{T}(U)$ and $\lambda _{A}(e)$, we derive the likelihood $P(K_{e},N_{U}|p_{e})$ as follows.
\vspace{-5pt}
\begin{eqnarray*}
	&&P(K_{e},N_{U}|p_{e})\\
	&=&P\left( \{K_{e}\}_{e\in E_{H}},\{N_{U}\}_{U\subseteq
		V_{SI}}|\{p_{e}\}_{e\in E_{H}},\{\lambda _{T}(U)\}_{U\subseteq
		V_{SI}}\right)  \\
	&=&P\left( \{K_{e}\}_{e\in E_{H}}|\{p_{e}\}_{e\in E_{H}},\{N_{U},\lambda
	_{T}(U)\}_{U\subseteq V_{SI}}\right)  \\
	&&\times P\left( \{N_{U}\}_{U\subseteq V_{SI}}|\{\lambda
	_{T}(U)\}_{U\subseteq V_{SI}}\right).
\end{eqnarray*}%
It is assumed that $\{K_{e}\}$ is independent with $\{N_{U}\}$, and $Q(p_{e},\hat{p}_{e}^{(i-1)})$ can be derived as follows:
\vspace{-2pt}
\begin{equation*}
\sum_{e\in E_{H}}{\log P(K_{e}|N_{U},p_{e})}+\log P(\{N_{U}\}_{U\subseteq
	V_{SI}}|\{\lambda _{T}(e)\}_{U\subseteq V_{SI}}).
\vspace{-2pt}
\end{equation*}%
Since only the first term contains the hidden $K_{e}$, only this term varies in different iterations of the EMS algorithm, because $\{N_{U}\}_{U\subseteq V_{SI}}$ in the second term always can be derived by finding the maximum likelihood as follows. Let $p_{U,k}$ denote the probability that the source set $U$ exactly tries to activate the destination node $k$ times, i.e., $p_{U,k}=P\{N_{U}=k\}$. The log-likelihood of $\lambda _{T}$ is
\begin{eqnarray*}
	&&\sum_{k}p_{U,k}\ln (\frac{\lambda _{T}^{k}e^{-{\lambda _{T}}}}{k!})=\sum_{k}p_{U,k}(-\lambda _{T}+k\ln \lambda _{T}-\ln k!) \\
	&=&-\lambda _{T}+(\ln \lambda _{T})\sum_{k}kp_{U,k}-\sum_{k}p_{U,k}\ln
	k!. \\
	&&
\end{eqnarray*}%
We acquire the maximum likelihood by finding the derivative with regard to $\lambda _{T}$:
\begin{equation}
-1+\frac{1}{\lambda _{T}}\sum_{k}kp_{U,k}=0.
\end{equation}%
Thus, the maximum log-likelihood estimation of $\lambda
_{T}=\sum_{k}kp_{U,k}$, representing that the expected activation times
(i.e., $\hat{\lambda}_{T}(e))$ is $N_{U}$. Let $\mathcal{A}=\{(\mathtt{v},\tau )\}$
denote the action log set, where each log $(\mathtt{v},\tau )$ represents that $\mathtt{v}$ is
activated at time $\tau $. $N_{U}$ is calculated by scanning $\mathcal{A}$
and find the times that all the nodes in $\mathtt{U}$ are activated.

Afterward, we focus on the first term of $Q(p_{e},\hat{p}_{e}^{(i-1)})$. Let $p_{e,k}=P\{K_{e}=k\}$ denote the probability that the hyperedge $e$ exactly activates the destination node $k$ times. In E-step, we first find the expectation for $K_{e}$ as follows.
\begin{align*}
&\sum_{e\in E_{H}}{\sum_{k=1,\cdots ,N_{U}}{p_{e,k}\log (k|N_{U},p_{e})}}\\
&=\sum_{k=1,\cdots ,N_{U}}{p_{e,k}\log P(k|N_{U},p_{e})} \\
&=\sum_{k=1,\cdots ,N_{U}}p_{e,k}\log ({\binom{N_{U}}{k}}{p_{e}}%
^{k}(1-p_{e})^{N_{U}-k}) \\
&= \sum_{k=1,\cdots ,N_{U}}{p_{e,k}\left( \log {\binom{N_{U}}{k}}+k\log
	p_{e}+(N_{U}-k)\log (1-p_{e})\right) }. \\
&
\vspace{-3mm}
\end{align*}

Since $\sum_{k=1,\cdots ,N_{U}}{p_{e,k}k}=E[K_{e}]$ and $\sum_{k=1,\cdots,N_{U}}{p_{e,k}}=1$, the log-likelihood of the first term is further simplified as
\begin{multline*}
\sum_{k=1,\cdots ,N_{U}}p_{e,k}\log {\binom{N_{U}}{k}}+N_{U}\log (1-p_{e})\\+E[K_{e})](\log p_{e}-\log (1-p_{e})).
\end{multline*}%
Afterward, M-step maximizes the Q-function by finding the derivative with regard to $p_{e}$:
\begin{eqnarray*}
	&&\frac{-N_{U}}{1-p_{e}}+E[K_{e})](\frac{1}{p_{e}}+\frac{1}{1-p_{e}})=0 \\
	&&p_{e}=E[K_e]/N_{U}
\end{eqnarray*}%
Therefore, the maximum likelihood estimator $\hat{p}_{e}$ is $\frac{E[K_{e}]}{N_{U}}$, $\hat{\lambda}_{T}(U)$ is $N_{U}$, and $\hat{\lambda}_{A}(e)=E[K_{e}]$.

The problem remaining is to take expectation of the latent variables $%
\{K_{e}\}$ in E-step. Let $\{w_{e,a}\}_{e\in E_{H},a=(\mathtt{v},\tau )\in \mathcal{A}%
}$ be the conditional probability that $v$ is activated by the source set $U$
of $e$ at $\tau$ given $v$ is activated at $\tau$, and let $E_{a}$
denote the set of candidate hyperedges containing every possible $e$
with its source set activated at time $\tau -1$, i.e., $E_{a}=%
\{(\mathtt{u}_{1},\mathtt{u}_{2},\cdots ,\mathtt{u}_{n})\rightarrow v|\forall i=1,\cdots ,n,\mathtt{u}_{i}\in
N(\mathtt{v}_{i}),(\mathtt{u}_{i},\tau -1)\in \mathcal{A}\}$. It's easy to show that given the
estimation of the probability of hyperedges, $w_{e,a}=\frac{\hat{p}_{e}}{%
	1-\prod_{e^{\prime}\in E_{a}}{(1-\hat{p}_{e^{\prime}})}}$, since $%
1-\prod_{e\in E_{a}}{(1-\hat{p}_{e^{\prime}})}$ is the probability for $\mathtt{v}$
to be activated by any hyperedge at time $\tau$. The expectation of $K_{e}$
is $\sum_{a\in A,e\in E_{a}\cap E_{H,n}}{w_{e,a}}$, i.e., the sum of
expectation of each successful activation of $\mathtt{v}$ from hyperedge $e$, and $%
E_{H,n}=\{(\mathtt{u}_1,\mathtt{u}_2,\cdots,\mathtt{u}_n;v)\in E_H\}$ contains all size $n$ hyperedges.

To address the data sparsity problem, we leverage information from similar hyperedges (described later). Therefore, our framework includes an additional step to smooth the results of M-Step. Kernel smoothing is employed in S-Step.

\textbf{In summary, we have the following steps:}

\textbf{E-Step}:
	\begin{align*}
	E[K_{e}]&=\sum_{a\in A,e\in E_{a}\cap E_{H,n}}{w_{e,a}}, \\
	w_{e,a}&=\frac{\hat{p}_{e}}{1-\prod_{e^{\prime}\in E_{a}}{(1-\hat{p}%
			_{e^{\prime}})}}.
	\end{align*}%

\textbf{\textbf{M-Step}:
	\begin{align*}
	p_{e} &= \frac{\sum_{a\in A,e\in E_{a}\cap E_{H,n}}{w_{e,a}}}{N_{U}}, \\
	\lambda_{A}\left( e\right) &= \sum_{a\in A,e\in E_{a}\cap E_{H,n}}{w_{e,a}}%
	, \\
	\lambda_{T}\left( U\right) &= N_{U}.
	\end{align*}%
}

\begin{table*}[t]
	\centering
	\caption{Comparison of precision, recall, and F1 for three models on Douban, Gowalla, Epinions}
	\begin{tabular}{|l||ccc||ccc|ccc|}
		\hline
		Dataset&  \textbf{Douban} &  &  & \textbf{Gowalla} & & & \textbf{Epinions}  & &\\ \hline
		& Precision & Recall & F1-Score & Precision & Recall & F1-Score & Precision & Recall & F1-Score\\ \hline
		GT & 0.420916 & 0.683275 & 0.520927& 0.124253 & 0.435963 & 0.171214 & 0.142565 & 0.403301 & 0.189999  \\
		IC & 0.448542 & 0.838615 & 0.584473 & 0.217694 & 0.579401 & 0.323537 & 0.172924 & 0.799560 & 0.247951 \\
		SIG & 0.869348 & 0.614971 & 0.761101 & 0.553444 & 0.746408 & 0.646652& 0.510118 & 0.775194 & 0.594529\\
		\hline
	\end{tabular}%
	\label{tab:Douban}
\end{table*}

\textbf{S-Step}: To address the data sparsity problem, we leverage information from similar hyperedges (described later). Therefore, in addition to E-Step and M-Step, EMS includes \textbf{S-Step}, which smooths the results of M-Step. Kernel smoothing is employed in S-Step as follows:
\begin{align*}
\hat{\lambda}_{A}\left( e\right) &=\sum_{a\in A,e^{\prime }\in E_{a}\cap
	E_{H,n}}{w_{e^{\prime },a}L_{h}\left( F(e)-F(e^{\prime })\right) } \\
\hat{\lambda}_{T}\left( U\right) &= \sum_{U\subseteq V_{SI}}{%
	N_{U}L_{h}\left( F(U)-F(U^{\prime })\right)}
\end{align*}%
where $L_{h}$ is a kernel function with bandwidth $h$, and $F$ is the mapping
function of hyperedges, i.e., $F(e)$ maps a hyperedge $e$ to a vector. The details of dimension reduction for calculating $F$ to efficiently map hyperedges into Euclidean space are shown in the next subsection.  If the hyperedges $e$ and $e^{\prime}$ are similar, the distance of the vectors
$F(e)$ and $F(e^{\prime})$ is small. Moreover, a kernel function $L_{h}(x)$ is a
positive function symmetric at zero which decreases when $|x|$ increases,
and the bandwidth $h$ controls the extent of auxiliary information taken
from similar hyperedges.\footnote{A symmetric Gaussian kernel function is often used \cite{HjortLfunction96}.} Intuitively, kernel smoothing can identify the correlation of $\hat{p}_{e_{1}}$ with $%
e_{1}=\mathtt{U}_{1}\rightarrow \mathtt{v}_{1}$ and $\hat{p}_{e_{2}}$ with $%
e_{2}=\mathtt{U}_{2}\rightarrow \mathtt{v}_{2}$ for nearby $\mathtt{v}_{1}$ and $\mathtt{v}_{2}$ and similar $\mathtt{U}_{1}$ and $\mathtt{U}_{2}$.

\subsection{Dimension Reduction}
To facilitate the computation of smoothing function in EMS algorithm, we exploit a dimension reduction technique \cite{Tenenbaum00, Yan07} to map a graph into Euclidean space as a set of vectors. Specifically, given $N$ users, let $Z\in \mathbb{R}^{N \times N}$ denote the projection matrix, where $z_i$ is the $i$-th row of $Z$ and is the projection of vertex $v_i$. We derive $Z$ as follows:
\begin{eqnarray*}
&&Z=\arg \min_{Z^{T}IZ=c}{z^{T}Lz}, \\
&&L=D-W,D_{ii}=\sum_{i\neq j}{W_{ij}}\forall i,
\end{eqnarray*}%
\noindent where $D$ is a diagonal matrix, $I$ is the identity matrix, and $L$ is the Laplacian matrix of distance matrix $W \in \mathbb{R}^{N \times N}$. This optimization problem attempts to preserve the distance between nodes. However, the constraint $Z^{T}IZ=c$ restricts that only $c$ columns of $Z$ can be non-zero vector. Therefore, the objective function reduces the dimension of $z_i$ from $N$ to $c$ while maintaining the local structure as much as possible.

Suppose we solve this generalized eigenproblem for first $c$ solutions and the $i$-th component of $j$-th eigenvector $z_j$ is denoted as $z_{j,i}$. The projection of $i$-th node $v_i$ has a $K$-dimensional representation $f(v_{i})=(z_{1,i},z_{2,i},\cdots,,z_{K,i})$. In our paper, we employ one of the most widely used nonlinear dimension reduction technique, ISOMAP \cite{Tenenbaum00} only. The distance matrix $W=-HSH/2$, where $H=I-1/N\overrightarrow{1}\overrightarrow{1}^T$ ($I$ is the identity matrix and $\overrightarrow{1}$ is the vector of all ones) and $S_{ij}$ is distance between two nodes in the graph. For our SIG, we first project the users whose graph is $G$, and then the items whose graph is set to be the complete graph. Other method sharing the above optimization formulation can be used as a substitute without much effort. 

By employing the above graph embedding approaches, we project the graph into Euclidean space while preserving the distance between the nodes locally. Therefore, the customers who are socially near and the similar commodity items can be extracted efficiently. Therefore, after the dimension reduction procedure, each node has a $K$-dimensional representation. Therefore, each hyperedge $e$ of size $n$ comprising of $n$ source nodes and $1$ destination node can be mapped to a vector on the space $\mathbb{R}^{(n+1)K}$.

\section{Evaluations}
\label{sec:experiment}
We conduct comprehensive experiments to evaluate the proposed SIG model, learning framework and seed selection algorithms. In Section \ref{expsetting}, we discuss the data preparation for our evaluation. In Section \ref{subsec:model}, we compare the predictive power of the SIG model against two baseline models: i) independent
cascade (IC) model learned by implementing \cite{Saito08KES} and ii) the generalized threshold (GT) model learned by \cite{Goyal11ICDM}.\footnote{\url{http://people.cs.ubc.ca/~welu/downloads.html}} In addition, we evaluate the learning framework based on the proposed EM and EMS algorithms. Next, in Section \ref{subsec:effectiveness}, we evaluate the proposed HAG algorithm for SIMP in comparison to a number of baseline strategies, including random, single node selection, social, and item approaches. Finally, in Section \ref{subsec:efficiency}, we evaluate alternative approaches for diffusion processing, which is essential and critical for HAG, based on SIG-index, Monte Carlo simulations and sorting enhancement.

\subsection{Data Preparation}
\label{expsetting}
Here, we conduct comprehensive experiments using three real datasets to evaluate the proposed ideas and algorithms. The first dataset comes from Douban \cite{DoubanData}, a social networking website allowing users to share music and books with friends. Dataset \textit{Douban} contains $5,520,243$ users and $86,343,003$ friendship links, together with $7,545,432$ (user, music) and $14,050,265$ (user, bookmark) pairs, representing the music noted and the bookmarks noted by each user, respectively. We treat those (user, music) and (user, bookmark) pairs as purchase actions. In addition to \textit{Douban}, we adopt two public datasets, i.e., \textit{Gowalla} and \textit{Epinions}. Dataset \textit{Gowalla} contains $196,591$ users, $950,327$ links, and $6,442,890$ check-ins \cite{Cho11KDD}. Dataset \textit{Epinions} contains $22,166$
users, $335,813$ links, $27$ categories of items, and $922,267$ ratings with timestamp \cite{Dataset3}. Notice that we do not have data directly reflecting item inferences in online stores, so we use the purchase logs for learning and evaluations. The experiments are implemented in an HP DL580 server with 4 Intel Xeon E7-4870 2.4 GHz CPUs and 1 TB RAM.

We split all three datasets into 5-fold, choose one subsample as training data, and test the models on the remaining subsamples. Specifically, we ignore the cases when the user and her friends did not buy anything. Finally, to evaluate the effectiveness of the proposed SIG model (and the learning approaches), we obtain the purchase actions in the following cases as the ground truth: 1) item inference - a user buys some items within a short period of time; and 2) a user buys an item after at least one of her friends bought the item. The considered periods of item inference and social influence are set differently according to \cite{itemperiod} and \cite{socialperiod}, respectively. It is worth noting that only the hyperedges with the probability larger than a threshold parameter $\theta$ are considered. We empirically tune $\theta$ to obtain the default setting based on optimal F1-Score. Similarly, the threshold parameter $\theta$ for the GT model is obtained empirically.
The reported precision, recall, and F1 are the average of these tests. Since both SIGs and the independent cascade model require successive data, we split the datasets into continuous subsamples.


\subsection{Model Evaluation}

\label{subsec:model}

%

Tables \ref{tab:Douban} present the precision, recall, and F1 of SIG, IC and GT on \textit{Douban}, \textit{Gowalla}, and \textit{Epinions}. All three models predict most accurately on \textit{Douban} due to the large sample size. The SIG model significantly outperforms the other two models on all three datasets, because it takes into account both effects of social influence and item inference, while the baseline models only consider the social influence. The difference of F1 score between SIG and baselines is more significant on \textit{Douban}, because it contains more items. Thus, item influence plays a more important role. Also, when the user size increases, SIG is able to extract more social influence information leading to better performance than the baselines. The offline training time is 1.68, 1.28, and 4.05 hours on \textit{Epinions}, \textit{Gowalla}, \textit{Douban}, respectively.
\begin{figure}[t]
	\centering
	\subfigure[Precision] {\includegraphics[width=1.3 in]{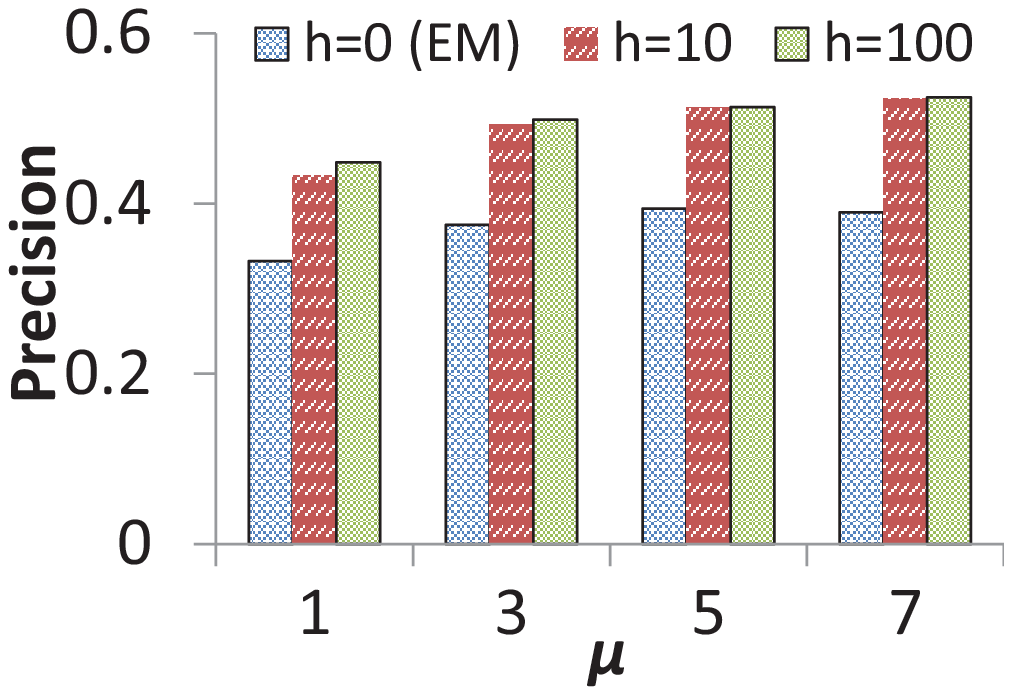}}
	\subfigure[F1-Score]{\includegraphics[width=1.3 in]{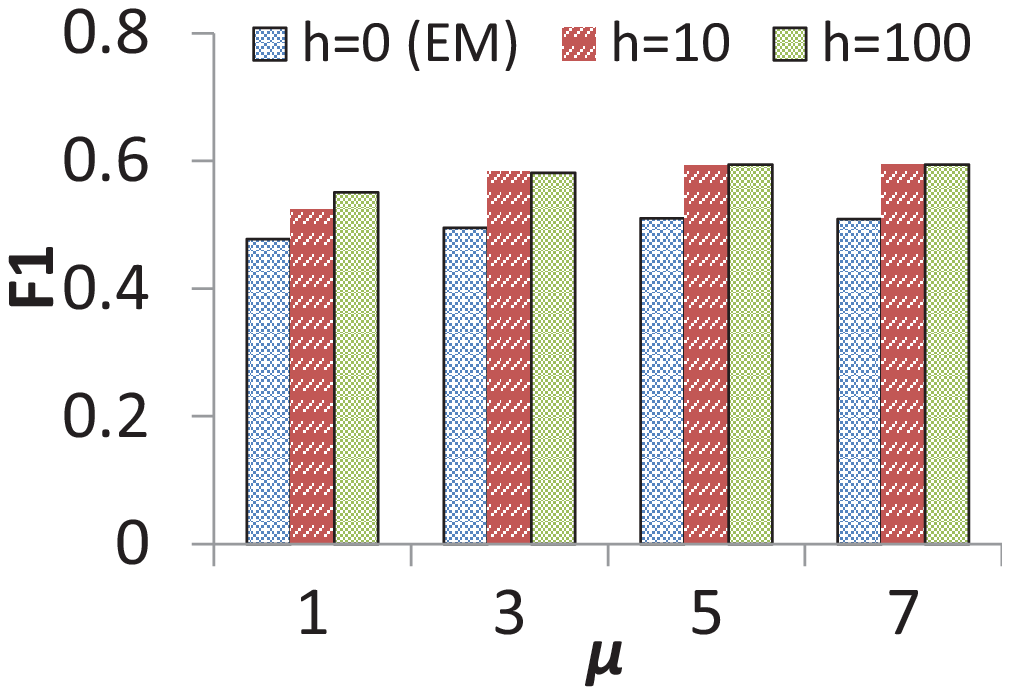}} \vspace{-5pt}\\
	\vspace{-8pt}
	\caption{Comparisons of precision and F1 in various $\mu$ and $h$ on \textit{Epinions}.}
	\label{fig:PrecF1}
\end{figure}

To evaluate the approaches adopted to learn the activation probabilities of hyperedges for construction of SIG, Fig. \ref{fig:PrecF1} compares the prevision and F1 of EMS and EM algorithms on Epinions (results on other datasets are consistent and thus not shown due to space limitation). Note that EM is a special case of EMS (with the smoothing parameter $h=0$, i.e., no similar hyperedge used for smoothing). EMS outperforms EM on both precision and F1-score in all settings of $\mu$ (the maximum size of hyperedges) and $h$ tested. Moreover, the precision and F1-score both increases with $h$ as a larger $h$ overcomes data sparsity significantly. As $\mu$ increases, more combinations of social influence and item inference can be captured. Therefore, the experiments show that a higher $\mu$ improves F1-score without degrades the precision. It manifests that the learned hyperedges are effective for predicting triggered purchases.


\subsection{Algorithm Effectiveness and Efficiency}

\label{subsec:effectiveness}

We evaluate HAG proposed for SIMP, by selecting top 10
items as the marketing items to measure their total
adoption, in comparison with a number of baselines:
1) \emph{Random approach (RAN)}. It randomly selects $k$ nodes as seeds.
Note that the reported values are the average of 50 random seed
sets.
2) \emph{Single node selection approach (SNS)}. It selects a node with
the largest increment of the total adoption in each iteration, until $k$
seeds are selected, which is widely employed in conventional seed
selection problem \cite{IC1,Chen10KDD,Kempe03KDD}. 3) \emph{Social
	approach (SOC)}. It only considers the social influence in selecting the $k$
seeds. The hyperedges with nodes from different products are eliminated
in the seed selection process, but they are restored for
calculation of the final total adoption. 4) \emph{Item approach (IOC)}. The
seed set is the same as HAG, but the prediction is based on item inference only.
For each seed set selected by the above approaches, the diffusion process is simulated 300 times.
We report the average in-degree of nodes learned from the three datasets in the following: Douban is 39.56; Gowalla is 9.90; Epinions is 14.04. In this section, we evaluate HAG by varying the number of seeds (i.e., $k$) using two metrics: 1)
total adoption, and 2) running time.

To understand the effectiveness, we first compared all those approaches with the optimal solution (denoted as OPT) in a small subgraph sampled, \emph{Sample}, from the SIG of \emph{Douban} with 50 nodes and 58 hyperedges.  Figures~\ref{fig:k_opt} (a) displays the total adoption obtained by different approaches. As shown, HAG performs much better than the baselines and achieves comparable total adoption with OPT (the difference decreases with increased k). Note that OPT is not scalable as shown in Figures~\ref{fig:k_opt} (b) since it needs to examine all combination with $k$ nodes.
Also, OPT takes more than 1 day for selecting 6 seeds in \emph{Sample}.
Thus, for the rest of experiments, we exclude OPT.


Figures~\ref{fig:k_spread} (a)-(c) compare the total adoptions of different
approaches in the SIG learnt from real networks. They all grow as $k$ increases, since a larger
$k$ increases the chance for seeds to influence others to adopt
items. Figure~\ref{fig:k_spread} (a)-(c) manifest that HAG outperforms all the other baselines for any $k$ in SIG model.
Among them, SOC fails to find good solutions since item inference is not
examined during seed selection. IOC performs poorly without considering social influence.
SNS only includes one seed at a time
without considering the combination of nodes that may activate many other nodes via hyperedges.
\begin{figure}[t]
	\centering
	\subfigure[\emph{Sample}] {\includegraphics[width=1.3 in]{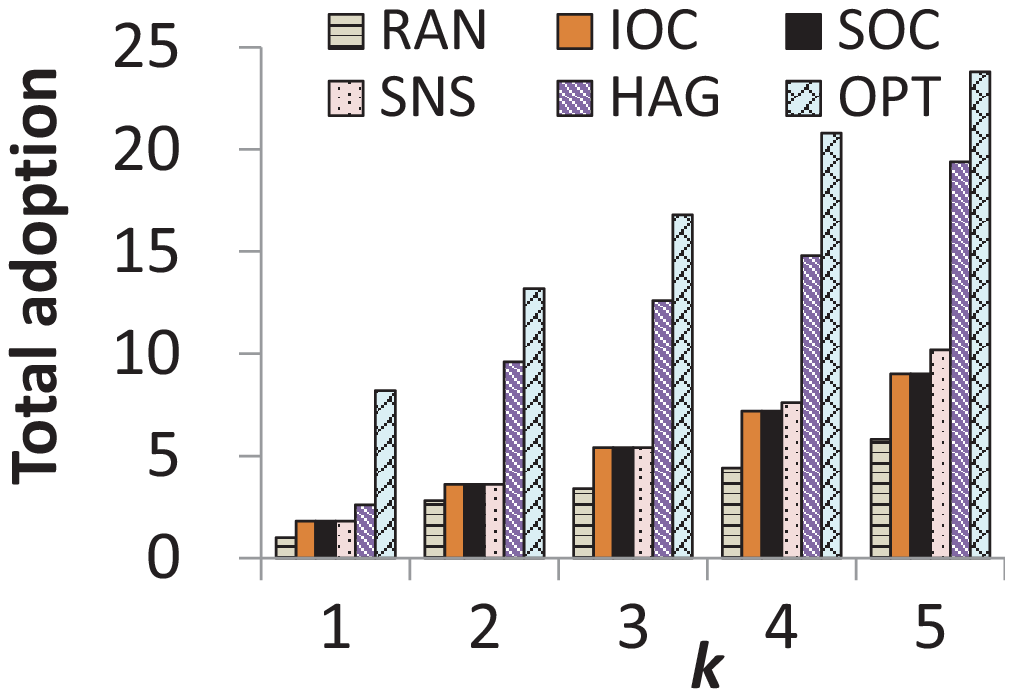}}
	\subfigure[Time]{\includegraphics[width=1.3 in]{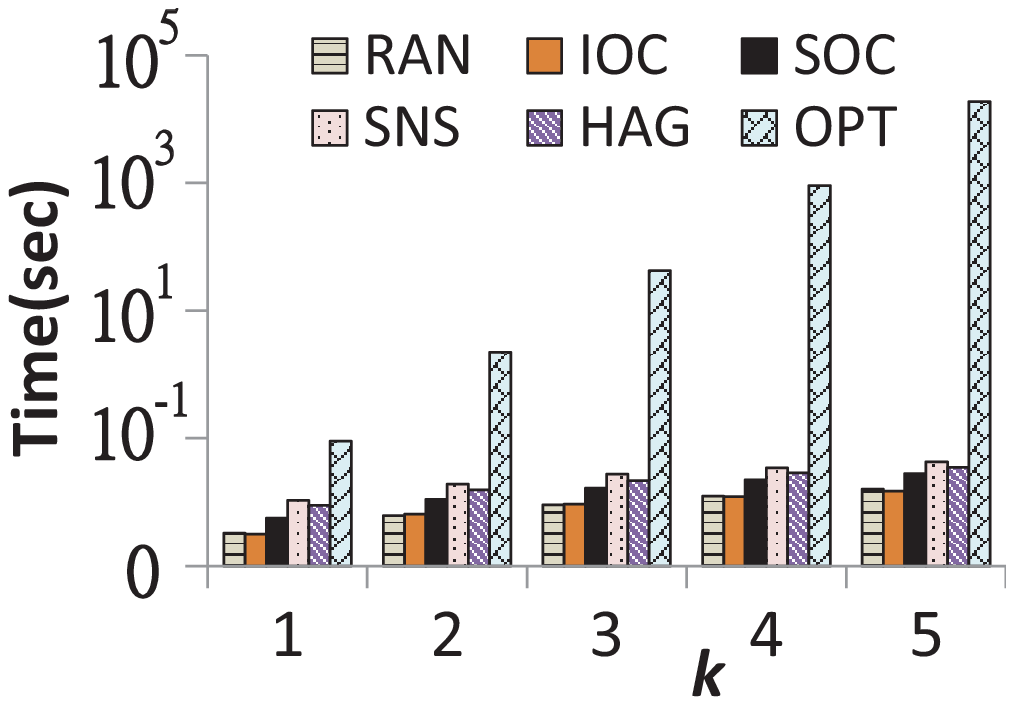}} \vspace{-5pt}\\
	\vspace{-8pt}
	\caption{Total adopting and running time of \emph{Sample} in various $k$}
	\label{fig:k_opt}
\end{figure}
\begin{figure}[t]
	\centering
	\subfigure[$\alpha_{G_{SI}}$(\emph{Douban})] {\includegraphics[width=1.3 in]{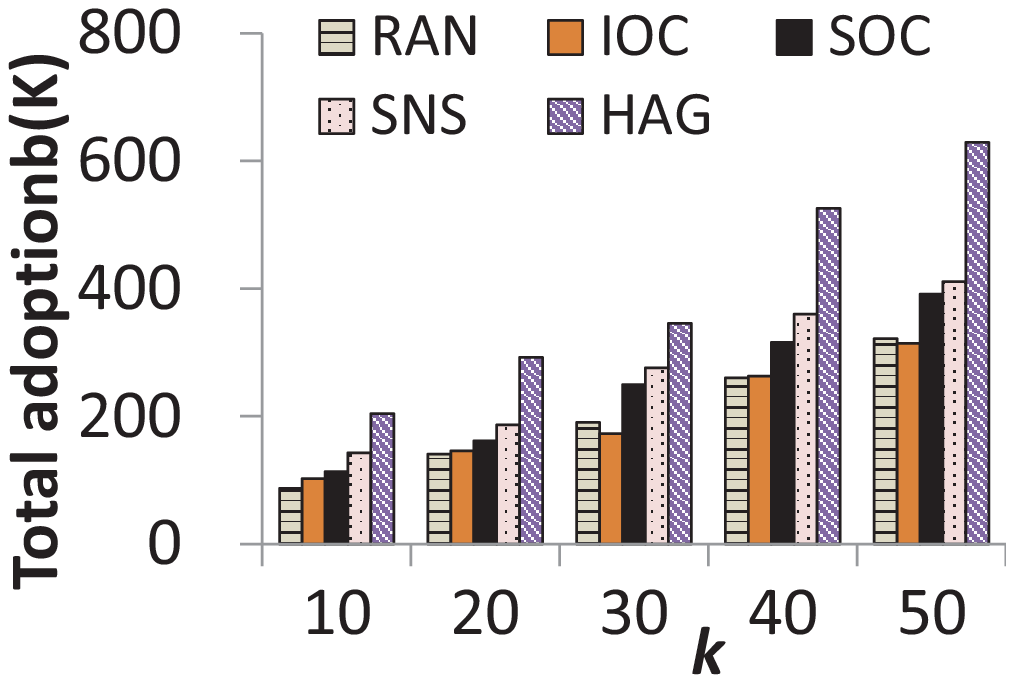}}
	\subfigure[$\alpha_{G_{SI}}$(\emph{Gowalla})] {\includegraphics[width=1.3 in]{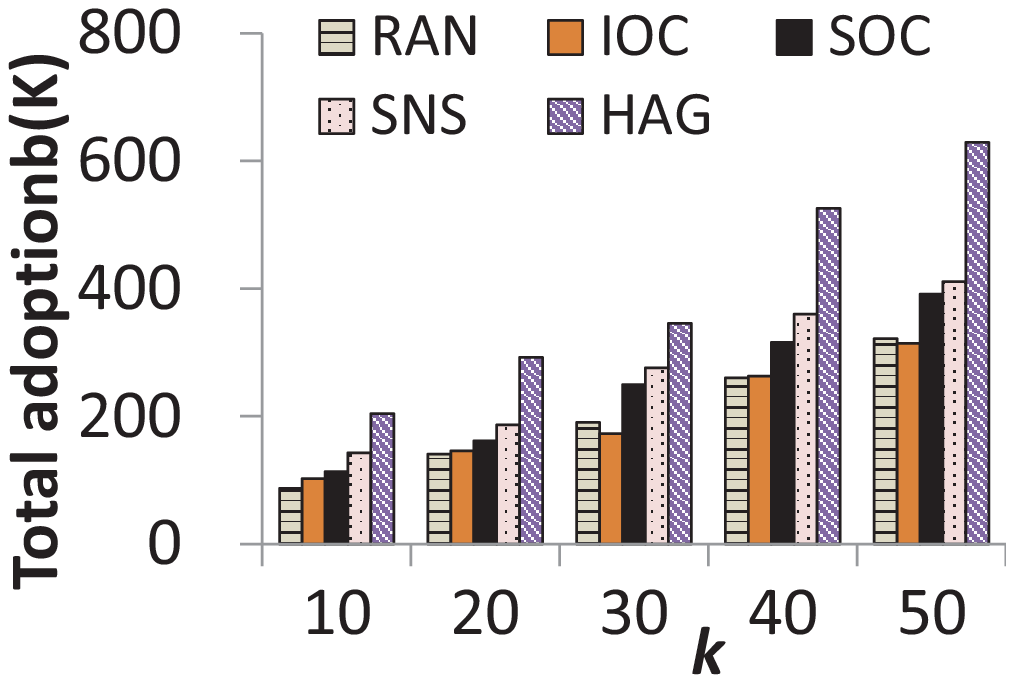}} \vspace{-5pt}\\
	\subfigure[$\alpha_{G_{SI}}$(\emph{Epinions})]{\includegraphics[width=1.3 in]{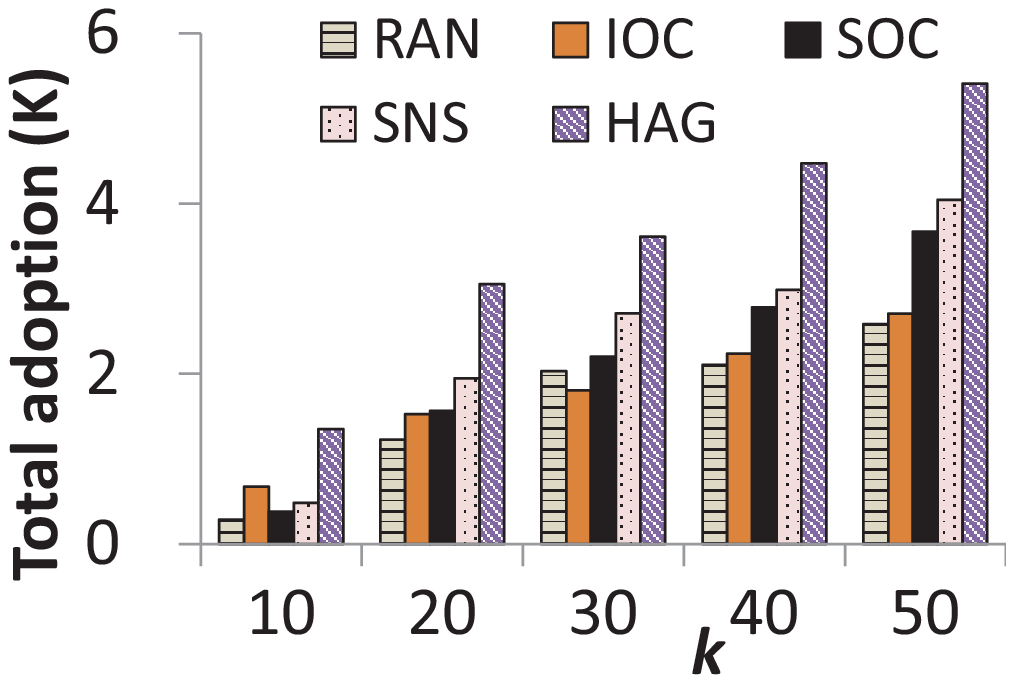}}
	\subfigure[Time (\emph{Douban})]{\includegraphics[width=1.3 in]{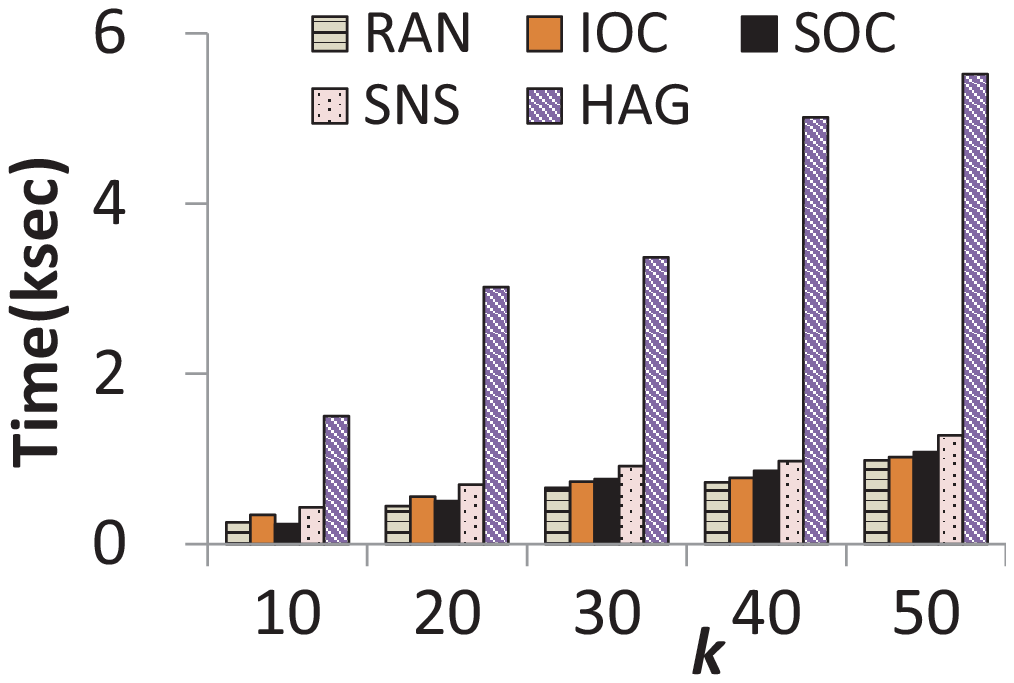}} \vspace{-5pt}\\
	\vspace{-8pt}
	\caption{Total adoption $\alpha_{G_{SI}}$ and running time in various $k$}
	\label{fig:k_spread}
\end{figure}



Figure~\ref{fig:k_spread} (d) reports the running time of
those approaches. Note that the trends upon \emph{Gowalla} and \emph{Epinions} are similar with \emph{Douban}. Thus we only report the running time of \emph{Douban} due to the space constraint. Taking the source combinations into account, HAG examines source combinations of hyperedges in $E_H$ and obtains a better solution by spending more time since the number of hyperedges is often much higher than the number of nodes.\footnote{To further reduce the running time, we eliminate the source combinations with little gain in previous iterations.} 

\subsection{Online Diffusion Processing\label{subsec:efficiency}}

\begin{figure}[t]
	\subfigure[Douban] {\includegraphics[width=1.05 in]{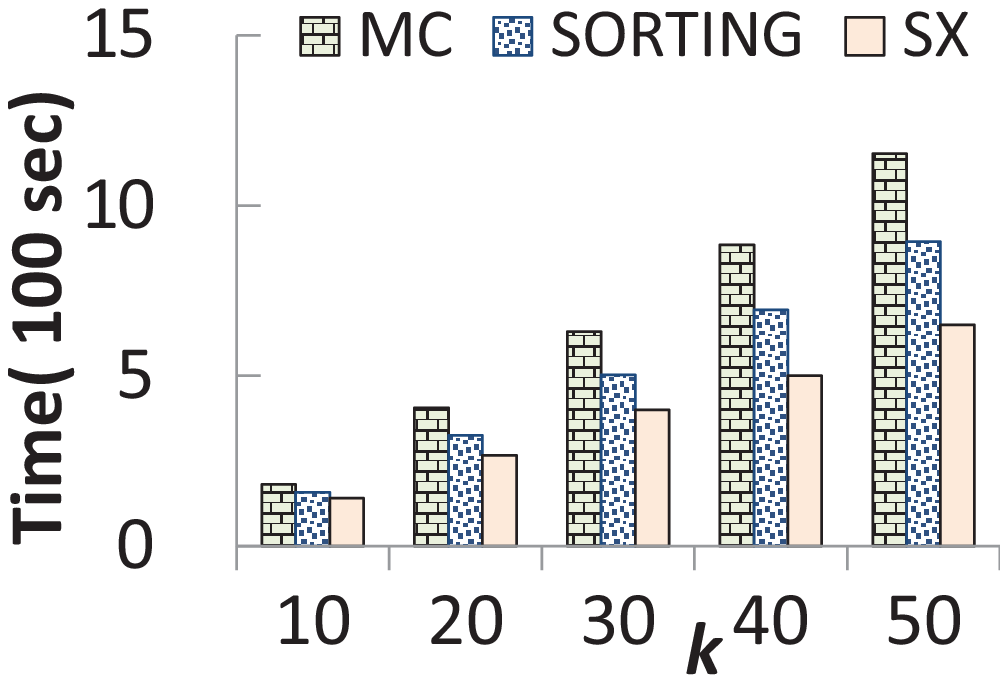}}
	\subfigure[Gowalla]{\includegraphics[width=1.05 in]{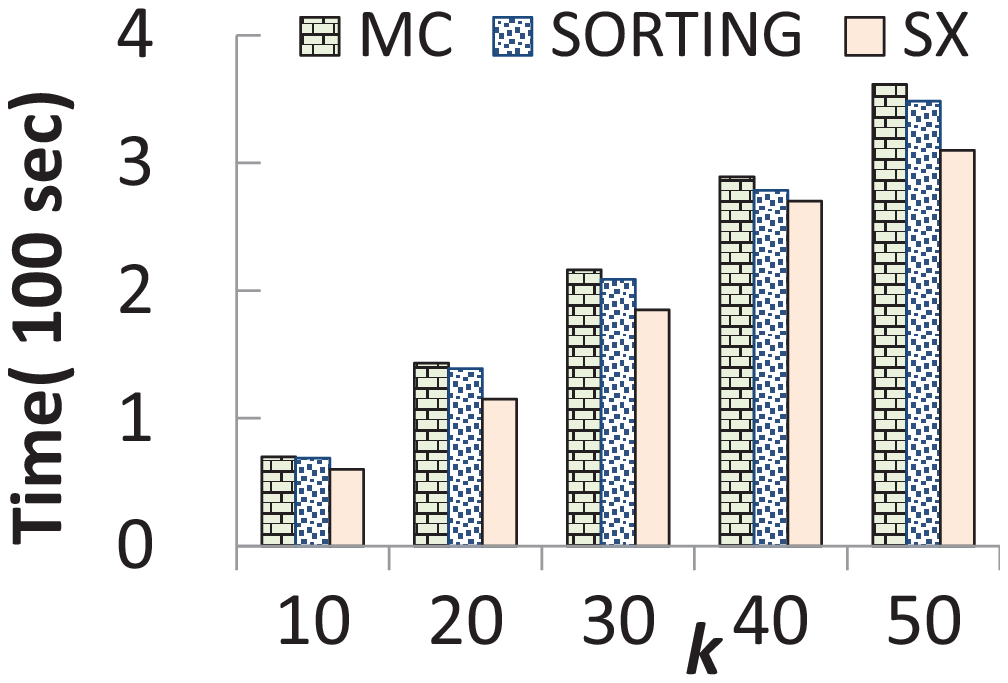}}
	\subfigure[Epinions]{\includegraphics[width=1.05 in]{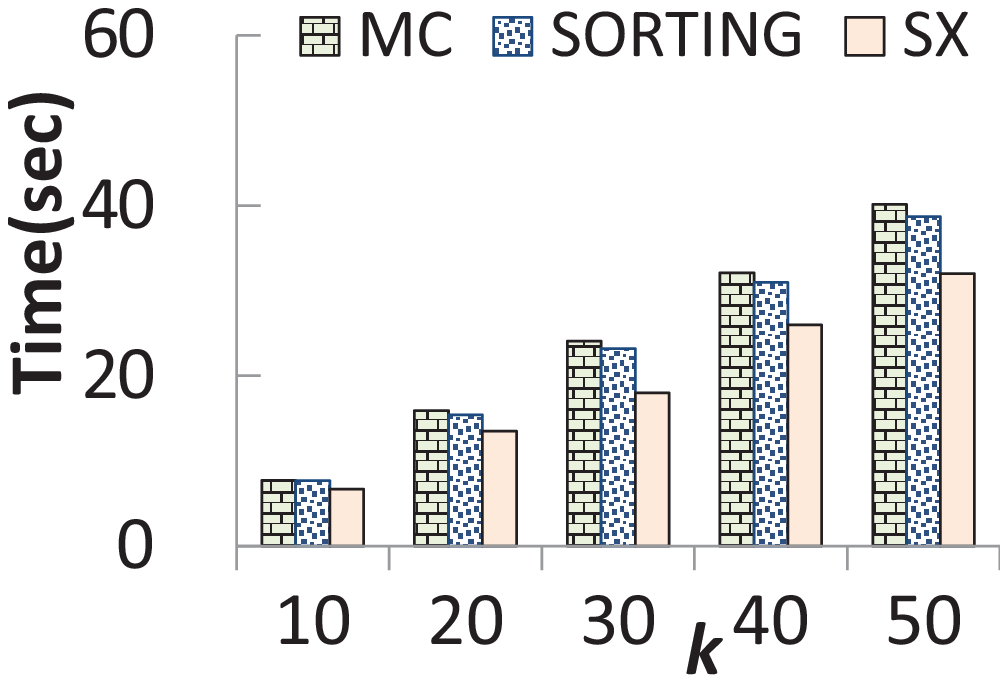}}
	\vspace{-8pt}
	\caption{Running time of different simulation methods}
	\label{fig:k_lattice}
\end{figure}

Diffusion processing is an essential operation in HAG. We evaluate the efficiency of diffusion processing based on SIG-index (denoted as SX), in terms of the running time,
in comparison with that based on the original Monte Carlo simulation (denoted as MC) and
the sorting enhancement (denoted as SORTING), which accesses the hyperedges in
descending order of their weights.
Figure~\ref{fig:k_lattice} plots the running time of SX, SORTING, and
MC under various $k$ using the \textit{Douban}, \textit{Gowalla}, and \textit{Epinions}.
For each $k$, the average running times of 50 randomly selected
seed sets for SX, SORTING, and MC, are reported. The diffusion process is simulated 300 times for each seed set.
As Figure~\ref{fig:k_lattice} depicts, the running time for all the three
approaches grows as $k$ increases, because a larger number of seeds
tends to increas the chance for other nodes to be activated. Thus, it
needs more time to diffuse. Notice that SX takes much less time than
SORTING and MC, because SX avoids accessing hyperedges with no source nodes newly activated while calculating the activation probability. Moreover, the SIG-index is updated dynamically according to the activated nodes in diffusion process.
Also note that the improvement by MC over SORTING in \textit{Douban} is
more significant than that in \textit{Gowalla} and \textit{Epinions}, because the average
in-degree of nodes is much larger in \textit{Douban}.
Thus, activating a destination at an early stage can effectively avoid processing many hyperedges later.


\section{Conclusion}

\label{sec:conclusion}

In this paper, we argue that existing techniques for item inference
recommendation and seed selection need to jointly take social influence and item inference into consideration. We propose Social Item Graph
(SIG) for capturing purchase actions and predicting potential purchase
actions. We propose an effective machine learning approach to construct a SIG from
purchase action logs and learn hyperedge weights. We also develop efficient
algorithms to solve the new and challenging Social Item Maximization Problem
(SIMP) that effectively select seeds for marketing. Experimental results demonstrate the superiority of the SIG model over existing models and the effectiveness and efficiency of the proposed algorithms for processing SIMP. We also plan to further accelerate the diffusion process by indexing additional information on SIG-index.



\bibliographystyle{abbrv}
\bibliography{reference}

\end{document}